\documentclass[
  aip,
  jmp,
  amsmath,
  amssymb,
  reprint
]{revtex4-2}
\usepackage{microtype}[draft]
\usepackage[scr=euler]{mathalfa}
\usepackage{eucal}[mathscr]
\usepackage{stmaryrd}
\usepackage{tikz-cd}
\usepackage{amsmath, amssymb}
\usepackage{mathtools}
\usepackage{amsthm}

\usepackage{ebproof}
\usepackage{proof}

\usepackage{boxedminipage}

\usepackage[framemethod=TikZ]{mdframed}
\mdfsetup{%
middlelinecolor=red,
middlelinewidth=2pt,
backgroundcolor=red!10,
roundcorner=10pt}

\usepackage{tikz}
   \usetikzlibrary{shapes.geometric, intersections}

\usepackage{comonads}

\renewcommand{\Sg}{\Sigma}
\newcommand{\Real}{\mathbb{R}}
\newcommand{\Eff}{\mathbb{E}}
\newcommand{\Proj}{\mathbb{P}}
\newcommand{\EA}{\mathbf{EA}}
\renewcommand{\EM}{\mathbf{EM}}
\newcommand{\orth}{\, \bot \,}
\newcommand{\Lop}{\mathcal{L}_{+}}
\newcommand{\povm}{\textnormal{\textsc{povm}}}
\newcommand{\pvm}{\textnormal{\textsc{pvm}}}
\newcommand{\XX}{\mathbb{X}}
\newcommand{\Mod}{\mathbf{Mod}}
\newcommand{\St}{\mathfrak{S}}
\newcommand{\OV}{\mathbf{OVect_u}}
\newcommand{\CF}{\mathrm{CF}}
\newcommand{\NCF}{\mathrm{NCF}}
\newcommand{\NCFobs}{\mathrm{NCF}_{\mathsf{obs}}}

\newcommand{\CFobs}{\mathrm{CF}_{\mathsf{obs}}}
\newcommand{\RR}{\mathcal{R}}

\begin{document}

\title{Effect-valued measurement models and contextuality}

\author{Samson Abramsky}
\email{s.abramsky@ucl.ac.uk}
\author{Lucas Stinchcombe}
\email{lucas.stinchcombe.25@ucl.ac.uk}
\affiliation{
Department of Computer Science,
University College London,
66--72 Gower St.,
London WC1E 6EA, U.K.
}

\begin{abstract}
We generalize the Abramsky--Brandenburger sheaf-theoretic treatment
of contextuality by replacing probability distributions with
distributions valued in a convex effect algebra \(A\). This yields a
notion of \(A\)-valued measurement model encompassing probabilistic,
deterministic, and quantum measurement structures within a single
framework.

States \(\sigma:A\to[0,1]\) induce ordinary empirical models,
allowing observable behaviour to be viewed as arising from
effect-valued measurement data. This leads to a distinction between
internal contextuality of \(A\)-models and observable contextuality
after state evaluation. We analyze the relationship between these
notions and identify coherence conditions under which observable
classical explanations assemble into internal ones.

Using the ordered-vector-space representation of effect modules, we
show that non-contextuality is characterized by feasibility of an
associated cone program, generalizing the linear-programming
formulation of contextuality in the probabilistic case. We also
introduce an effect-valued contextual fraction and study its relation
to observable contextuality witnesses. We show how cone duality leads to a notion of Bell witnesses for contextuality as a direct generalization of Bell inequalities.

Finally, we analyze sharp realizability and uniform dilation for
measurement models, clarifying the relationship between general POVMs
and projective measurements, and show how resource-indexed effect
structures induce graded monads generalizing the quantum monad.
\end{abstract}

\maketitle


\section{Introduction}

Quantum theory describes measurement outcomes via positive operator-valued measures (POVMs), whose effects capture the operational content of experiments. While projective measurements play a distinguished role, realistic quantum systems require general measurements, and the structural role of these in non-classical phenomena such as contextuality remains incompletely understood.

At the same time, sheaf-theoretic approaches to contextuality \cite{AbramskyBrandenburger2011} provide a precise mathematical account of non-classicality in terms of the non-existence of global sections, together with associated cohomological obstructions. These methods capture the incompatibility of measurement outcomes across different contexts, but are typically formulated in probabilistic terms and do not directly accommodate general effect-valued measurements.

The aim of this paper is to develop a unified approach to quantum measurement in which contextuality, general measurements, and compositional structure can be analysed together. We introduce \emph{measurement models} defined on presheaves valued in effect algebras,  in which outcomes are assigned elements of an abstract effect structure, while states recover operational predictions via evaluation. This extends probabilistic models while retaining a direct connection to standard quantum formalism.

\medskip

\noindent
\textbf{Relation to existing operational frameworks.}
Generalised probabilistic theories (GPTs) \cite{barrett2007information,plavala2023general} and operational approaches to contextuality \cite{spekkens2005contextuality} provide powerful and widely used frameworks for describing states, effects, and measurements. In GPTs, the primary objects are convex state spaces and their dual effect spaces, while operational approaches, based on ontological models and operational equivalences, characterise contextuality via the (non-)existence of representations respecting specified equivalences, and naturally accommodate unsharp measurements.

The perspective adopted here is complementary. Rather than focusing on state spaces or ontological representations, we take as fundamental the structure of measurement scenarios themselves, including compatibility relations between measurements and the problem of extending locally defined outcomes across overlapping contexts. This leads to a sheaf-theoretic and obstruction-theoretic formulation, in which contextuality appears as a failure of global consistency.

In this sense, GPTs and operational approaches describe what measurements and states are, while the present framework analyses how measurement outcomes can or cannot be consistently combined across incompatible contexts. From a mathematical viewpoint, this involves presheaf constructions over measurement covers, structures which are not captured by convex or affine structures alone. In particular, contextuality is treated here as a structural obstruction independent of any choice of ontological model or equivalence structure.

\medskip

From the perspective of the sheaf-theoretic approach to contextuality, the key modification to the functorial architecture is to generalize from the usual distribution monads valued in semirings to distribution functors valued in effect algebras.
This leads to the generalization from the usual  probabilistic empirical models to a broader class of \emph{measurement models}. The connection with empirical models is retained via the states of the effect algebra, which induce mappings from measurement models to empirical models, which can be regarded as observable windows on the behaviour of the measurement models. 

\medskip

\noindent
The theory developed here can be seen as integrating key features of these approaches. From the sheaf-theoretic perspective, it inherits the analysis of contextuality in terms of compatibility structure and obstruction to global sections. From generalised probabilistic theories, it incorporates the use of effect structures and states to capture the operational content of measurements. From operational approaches based on equivalences, it takes the emphasis on general (including unsharp) measurements. These elements are brought together within a single setting in which contextuality, general measurements, convex structure and compositional behaviour can be analysed in a unified way.

\medskip

\noindent
\textbf{Overview of results.}
We develop a theory of effect--valued measurement models
and establish a number of structural results.

First, we show that contextuality extends naturally to the
effect-valued setting, and that the standard
Abramsky--Brandenburger framework is recovered when
$A=[0,1]$. States $\sigma:A\to[0,1]$ induce observable windows
on $A$-models, leading to a distinction between internal
and observable contextuality.

A central result establishes that internal non-contextuality
is equivalent to the existence of coherent families of
observable classical explanations across states.

Using the ordered-vector-space representation of effect modules,
we characterize contextuality as feasibility of an associated
cone program, generalizing the linear-programming formulation
for probabilistic empirical models. We  extend the
contextual fraction to the effect-valued setting and show that
cone duality yields generalized contextuality witnesses extending
Bell inequalities.

We then analyze sharp realizability and uniform dilation for
measurement models, identifying structural obstructions to
globally compatible projective realizations.

Finally, we show that resource-indexed effect structures induce
graded monads generalizing the quantum monad on relational
structures, providing a compositional setting for contextuality,
non-local games, and quantum advantage.



\section{Background}

\subsection{The sheaf-theoretic formulation of contextuality}
\label{sheaf:subsec}

We briefly review the basics of the Abramsky-Brandenburger sheaf-theoretic approach to contextuality \cite{AbramskyBrandenburger2011}.

A \emph{measurement scenario} is a structure $(X, \Sg, O)$, where:
\begin{itemize}
\item $X$ is a set of \emph{measurements}.
\item $\Sg$ is an \emph{abstract simplicial complex} over $X$, \ie~a family of finite subsets of $X$ such that $C \subseteq C' \in \Sg$ implies $C \in \Sg$, and $\bigcup \Sg = X$. We refer to the elements of $\Sg$ as \emph{contexts}. They specify which measurements can be performed together. The downward closure property is evident from this interpretation.
\item $O = \{ O_x \}_{x \in X}$ specifies a set of \emph{outcomes} for each measurement.
\end{itemize}

A minor variant of this definition is often used, where instead of $\Sg$ we specify the set $\MM$ of maximal elements of $\Sg$, referred to as the \emph{measurement cover}. If $X$ is finite, as is commonly the case, this is equivalent to the above formulation.

Given a measurement scenario $(X, \Sg, O)$, 
the \emph{event sheaf}  $\ES : P^{\op} \to \Set$, where $P = (\Pow(X),{\subseteq})$,
is defined by $\ES(U) = \prod_{x \in U} O_x$, with the obvious action by projection: $\rho^{V}_{U} : \ES(V) \to \ES(U)$ when $U \subseteq V$.

\textbf{Notation} if $U \subseteq V$ and $s \in \ES(V)$, we write $s |_U := \rho^{V}_{U}(s)$. 

The elements of $\ES(U)$ are referred to as \emph{local sections} over $U$. They specify an outcome for each measurement $x 
\in U$. The sheaf property for $\ES$ is easily verified: given any compatible family of local sections $\{ s_{C} \}_{C \in \Sg}$, where compatibility means that $s_{C} = s_{C'} |_{C}$ when $C \subseteq C'$, there exists a unique global section $s \in \ES(X)$ such that $s |_{C} = s_{C}$ for each $C \in \Sg$. Thus at the deterministic level of events, we can glue locally consistent descriptions into a single globally consistent description.

Compatibility can be defined equivalently in terms of the cover $\MM$: a family $\{ s_{C} \}_{C \in \MM}$ is compatible if for all $C, C' \in \MM$, and $s \in \ES(C)$, $s' \in \ES(C')$, $s |_{C \cap C'} = s' |_{C \cap C'}$.

Contextuality arises when we pass from deterministic to indeterministic families of events.
To capture this, we introduce the \emph{distribution monad} $\Dist_{R}$, parameterised on a commutative semiring $R$.
We recall firstly that a commutative semiring $(R, +, 0, \cdot, 1)$ differs from a commutative ring  in that the additive part $(R, +, 0)$ is only required to be a commutative monoid rather than an abelian group. Thus an alternative name for semirings is ``rigs'' --- rings without negatives. Standard examples are the non-negative reals $\Real_{\geq 0}$, and the Boolean semiring $\{ 0,1\}$ with addition given by disjunction.

Given a semiring $R$, we define the functor $\Dist_{R} : \Set \to \Set$ by 
\[ \Dist_R(X) = \{ d : X \to R \mid \supp(d) \; \text{finite}, \; \sum_{x \in X} d(x) = 1 \} . \]
Here $\supp(d) := \{ x \in X \mid d(x) \neq 0 \}$ is the support of $d$.
Given $f : X \to Y$, $\Dist_{R}(f) : \Dist_{R}(X) \to \Dist_{R}(Y)$ is defined by $\Dist_{R}(f)(d)(y) = \sum_{f(x) = y} d(x)$, the push-forward of $d$ along $f$.
This functor extends to a monad on $\Set$, with unit $\eta_X : X \to \Dist_{R}(X)$ given by the Dirac delta distribution $\eta_X := x \mapsto \delta_x$. The monad multiplication is given by
$\mu_X : \Dist_{R} (\Dist_{R}(X)) \to \Dist_{R}(X)$, where $\mu_X(D)(x) := \sum_{d \in \Dist_{R}(X)} D(d) \cdot d(x)$.  This is well-defined by finiteness of the support of $D$.

A noteworthy point, to which we shall return, is that the multiplication of the semiring is only used in defining the monad multiplication. The underlying functor and the unit use only the additive structure.

We now define the central notion of the sheaf-theoretic approach, that of \emph{empirical model}.
An empirical model $e$ is a compatible family on the cover $\Sg$ over the presheaf $\Dist_{R} \ES : P^{\op} \to \Set$. Explicitly, $e = \{ e_{C} \}_{C \in \Sg}$, where $e_{C} \in \Dist_{R} (\ES(C))$. Compatibility is, as before: $e_{C} = e_{C'} |_{C}$ whenever $C \subseteq C'$. Note that restriction in the presheaf $\Dist_{R} \ES$ is \emph{marginalization}. Compatibility has a clear physical meaning, as a generalized No-Signalling or No-Disturbance property \cite{AbramskyBrandenburger2011}.
In functorial terms, a compatible family is a natural transformation $\One \natarrow \Dist_{R} \ES_{\Sg}$, where $\ES_{\Sg}$ is the restriction of $\ES$ to $\Sg$, and $\One$ is the constant presheaf on $\Sg$ with value a singleton.

Now contextuality can be defined as a property of empirical models. An empirical model $e$ is \emph{contextual} if there is no global section $d \in \Dist_{R}(\ES(X))$ such that $d |_{C} = e_{C}$ for all $C \in \Sg$. Thus contextual empirical models are witnesses to the failure of the sheaf property for the presheaf $\Dist_{R} \ES$. Otherwise put, they  witness  the failure of \emph{global consistency} of a family of \emph{locally consistent} data.

This notion of contextuality can be refined to a hierarchy of ``strengths'' of contextuality \cite{AbramskyBrandenburger2011}. It can be shown to give a unified account of a very wide range of phenomena in contextuality and non-locality, including probabilistic contextuality as in Bell's theorem, possibilistic arguments such as the Hardy paradox and the GHZ construction, All-versus-Nothing arguments, state-independent contextuality such as Kochen-Specker paradoxes and the Peres-Mermin magic square, etc. \cite{AbramskyBrandenburger2011,DBLP:conf/csl/AbramskyBKLM15}.

Further developments include quantifying contextuality using the contextual fraction \cite{AbramskyBarbosaMansfield2017}, developing a resource theory of contextuality \cite{AbramskyBarbosaMansfield2017,abramsky2019comonadic}, and giving cohomological characterizations of contextuality as obstructions to the existence of global sections \cite{DBLP:journals/corr/abs-1111-3620,DBLP:conf/csl/AbramskyBKLM15}.

\subsection{Effect algebras}

Effect algebras have been widely studied in quantum foundations \cite{FoulisBennett1994}.
They form an algebraic abstraction of quantum effects and unsharp measurements, which also includes sharp measurements and classical probability as special cases.
They are closely related through their linear representations to Generalized Probabilistic Theories (GPT's) \cite{barrett2007information}, and to the study of contextuality arising from unsharp measurements in the operational framework \cite{spekkens2005contextuality}.

\begin{definition}
An \emph{effect algebra} is a set $A$ equipped with a partial binary operation $\oplus : \orth \to A$ where $\orth \subseteq A^2$, and   elements $0, 1 \in A$, satisfying:
\begin{itemize}
\item If $a \orth b$ then  $b \orth a$, and  $a \oplus b = b \oplus a$.
\item $a \orth 0$ and $a \oplus 0 = a$ for all $a$.
\item If $a \orth b$ and $(a \oplus b) \orth c$, then $b \orth c$, $a \orth  (b \oplus c)$, and $(a \oplus b)\oplus c = a \oplus (b \oplus c)$.
\item For every $a$ there exists a unique element $a^\perp$ such that $a \oplus a^\perp = 1$.
\item If $a \orth 1$ , then $a=0$.
\end{itemize}
\end{definition}
The operation $a \mapsto a^\perp$ is \emph{orthosupplementation}.

Some immediate consequences of these definitions are that $1 = 0^\perp$, and that $a^{\perp\perp} = a$.
Also, the cancellation law holds: if $a \orth b$, $a \orth c$, and $a \oplus b = a \oplus c$, then $b =c$. Hence if we define $a \leq b \equiv \exists c. \, a \oplus c = b$, then this forms a partial order on $A$, with $0$ as the least element, and $1$ the greatest element. Furthermore, $a \orth b$ iff $a \leq b^{\perp}$, and if $a \leq a' \orth b$, then $a \orth b$.

Given a finite set $S \subseteq A$, we define $\bigoplus_{a \in S} a$ as follows:
\begin{itemize}
\item Choose an ordering $S = \{ a_1, \ldots , a_n \}$.
\item Define $s_0 := 0$, $s_{i+1} := s_i \oplus a_{i+1}$ if $s_i$ defined and $s_i \orth a_{i+1}$, undefined otherwise.
\end{itemize}
Then $\bigoplus_{a \in S} a = s_n$ if $s_n$ is defined, and undefined otherwise.
\begin{proposition}
The value of $\bigoplus_{a \in S} a$ is independent of the order chosen on $S$.
\end{proposition}
\begin{proof}
This follows from the partial commutative monoid laws.
\end{proof}
We then define an equation $\bigoplus_{a \in S} a = b$ to mean that $\bigoplus_{a \in S} a$ is defined and equal to $b$.
This extends to families $\{ a_i \}_{i \in I}$, which may be infinite as long as $\{ i \in I \mid a_i \neq 0 \}$ is finite.

Key examples of effect algebras:
\begin{itemize}
\item The unit interval $[0, 1]$, with $r \orth s \equiv r+s \leq 1$. This is the locus for classical probability. Note that the partial order in $[0,1]$ forms a complete bounded lattice.
More generally, we have powers $[0,1]^X$,  which are also lattice-ordered.
\item  Boolean algebras are effect algebras, with $a \orth b \, \equiv \, a \wedge b = 0$, $a \oplus b = a \vee b$, and $a^\perp = \neg a$. 
Note that disjunction is only defined in the disjoint case.
\item Given a Hilbert space $\HH$, we write $\Lop^S(\HH)$ for the positive 
operators on $\HH$ with eigenvalues in $S \subseteq [0, \infty)$. The \emph{quantum effects} are $\Eff(\HH) := \Lop^{[0,1]}(\HH)$, and form an effect algebra, with $P \orth Q$ iff $P + Q \in \Lop^{[0,1]}(\HH)$.
This is the key motivating example from quantum theory. Note that 
\[ \Eff(\HH) = \{ P \in \Lop^{[0, \infty]}(\HH) \mid \mathbf 0 \leq P \leq I \} , \]
where the ordering on positive operators is defined by $P \leq Q \equiv Q - P \; \text{positive}$.
\item The \emph{projections} on $\HH$ are $\Proj(\HH) := \Lop^{\{ 0,1 \}}(\HH)$. These are the idempotent positive operators. They form a sub-effect algebra of $\Eff(\HH)$, with 
\[ P \orth Q \; \iff \; PQ = \Zero \; \iff \; P + Q \in \Proj(\HH) \; \iff \; P + Q \in \Eff(\HH). \]
Physically, they correspond to \emph{sharp} measurements, while general effects allow for \emph{unsharp} measurements.
\end{itemize}

A morphism $f : A \to B$ of effect algebras is a function $f$ such that 
\begin{itemize}
\item $f(1) = 1$.
\item If $a \orth b$ then $f(a) \orth f(b)$, and $f(a \oplus b) = f(a) \oplus f(b)$.
\end{itemize}
It follows automatically that morphisms preserve $0$ and orthosupplementation.

We obtain a category $\EA$ of effect algebras and morphisms.

\subsubsection{Effect modules}

We are interested in effect algebras with additional structure. 

\begin{definition}
An \emph{effect module} (or \emph{convex effect algebra}) is an effect algebra $A$ equipped with a scalar multiplication
\[
[0,1] \times A \to A, \qquad (r,a)\mapsto r\cdot a,
\]
such that for all $r,s\in[0,1]$ and $a,b\in A$:
\begin{itemize}
\item $1\cdot a = a$, \quad $0\cdot a = 0$,
\item $(rs)\cdot a = r\cdot (s\cdot a)$,
\item if $a \orth b$, then
\[
r\cdot(a \oplus b) = (r\cdot a)\oplus (r\cdot b),
\]
\item if $r+s \leq 1$, then
\[
(r+s)\cdot a = (r\cdot a)\oplus (s\cdot a).
\]
\end{itemize}
\end{definition}

Effect modules provide exactly the structure needed to interpret convex combinations internally.

The effect algebras $[0,1]$ and $\Eff(\HH)$ equipped with the obvious notions of scalar multiplication are effect modules. By contrast, $\{ 0,1 \}$ and $\Proj(\HH)$ are not convex. For finite-dimensional $\HH$, $\Eff(\HH)$ is the convex closure of $\Proj(\HH)$. This extends in a limiting sense  to the infinite-dimensional case \cite{heinosaari2011mathematical}.

Morphisms of effect modules are effect algebra morphisms which preserve the scalar multiplication. This yields a category $\EM$.

\subsubsection{States}

A \emph{state} of an effect algebra $A$ is an effect algebra morphism $\sg : A \to [0,1]$.

In the case of the effect algebra of projections, we have the following refined version of Gleason's theorem \cite{gleason,dvurecenskij1993gleason}:
\begin{theorem}
\label{th:Gleason}
If $\dim(\HH) > 2$, there is a bijective correspondence between states $\sg$ of $\Proj(\HH)$ and density operators $\rho$ on $\HH$, via the Born rule:
\[ \sg(P) = \Tr(P \rho) . 
\]
\end{theorem}

For quantum effects, there is a stronger version due  to 
Busch \cite{Busch2003Gleason}.
\begin{theorem}
\label{th:Busch}
If $\dim(\HH) > 1$, there is a bijective correspondence between states $\sg$ of $\Eff(\HH)$ and density operators $\rho$ on $\HH$, via the Born rule:
\[ \sg(P) = \Tr(P \rho) . 
\]
\end{theorem}

\begin{remark}
Unlike the projection lattice, the effect algebra $\Eff(\HH)$
admits no dispersion-free states. The presence of convex structure forces
all states to be given by density operators, even in dimension $2$.
\end{remark}

\section{Effect-valued distributions}
\label{sec:eff-valueddists}

We now come to a key point in our development. Rather than regarding effect algebras or modules as primary structures in their own right, we shall view them as \emph{weights} for distributions, generalising the $R$-valued weights for semirings $R$ which we encountered in Section~\ref{sheaf:subsec}.

Given an effect algebra $A$, we define a functor $\Dist_A : \Set \to \Set$ in exactly the same manner as we did for $\Dist_R$. Explicitly, given a set $X$, we define
\[
\Dist_A(X)=\left\{d:X\to A \text{ with finite support} \;\middle|\; \bigoplus_{x\in X} d(x)=1\right\}.
\]
Given $f:X\to Y$, define $\Dist_A(f)$ by pushforward:
\[
(\Dist_A(f)(d))(y)=\bigoplus_{x\in f^{-1}(y)} d(x).
\]
Using exactly the same arguments as for $\Dist_R$, it is routine to show:
\begin{proposition}
$\Dist_A : \Set \to \Set$ is a functor.
\end{proposition}

Just as for $\Dist_R$, there is a natural unit (Dirac map)
\[
\eta_X^A : X \to \Dist_A(X),
\qquad
\eta_X^A(x)(x') =
\begin{cases}
1 & x=x' \\
0 & \text{otherwise}.
\end{cases}
\]

We note that this construction was studied previously by Jacobs \cite{jacobs2011probabilities}, in order to define an adjunction between effect algebras and an abstract notion of convex functors.

There is an interesting comparison with the work of Uijlen and Staton \cite{staton2018effect}. They take effect algebras as generalizations of the event spaces of classical probability, while still taking the probability weights in $[0,1]$, as in classical probability. By contrast, we use ``locally classical'' event spaces,  as controlled by the event sheaf, while generalizing to effect-valued probability weights.

We examine what this construction produces for the  key examples of effect algebras:
\begin{itemize}
\item It is easy to see that $\Dist_{[0,1]}$ is isomorphic to the standard probability distribution functor $\Dist_{\Real_{\geq 0}}$.
\item Note that distributions in $\Dist_{\{ 0,1 \}}$ assign $1$ to exactly one argument, and hence correspond to deterministic assignments. Thus $\eta_X^{\{0,1\}} : X \xrightarrow{\cong} \Dist_{\{0,1\}}(X)$ is an isomorphism. 
\item For the quantum effects $\Eff(\HH)$, $\Dist_{\Eff(\HH)}(X)$ gives the discrete \povm's over $\HH$ with outcomes in $X$. 
\item For the projections $\Proj(\HH)$, $\Dist_{\Proj(\HH)}(X)$ gives the discrete \pvm's over $\HH$ with outcomes in $X$.
\end{itemize}

\subsection{Failure of monad structure}
\label{subsec:failmonad}

In general, the functor $\Dist_A$ does \emph{not} admit a monad structure. The obstruction arises in defining the multiplication
\[
\mu_X : \Dist_A(\Dist_A(X)) \to \Dist_A(X),
\]
which would ``flatten'' distributions of distributions.

A natural candidate would be
\[
\mu_X(\Phi)(x)
=
\bigoplus_{d \in \Dist_A(X)} \Phi(d) \cdot d(x),
\]
but this requires a multiplication
\[
\_ \cdot \_ :  A \times A \to A,
\]
which is not part of the structure of a general effect algebra or effect module.

Thus additional structure is required. \emph{Effect
monoids} \cite{jacobs2011probabilities} give one such setting: they are effect algebras equipped with
a multiplicative monoid operation distributing over the partial addition.

An example is provided by the effect algebra $[0,1]$, which is an effect monoid under multiplication, and hence gives rise to a monad on $\Set$ extending $\Dist_{[0,1]}$.
This is easily seen to be isomorphic to the usual probability distribution monad.

\subsection{Action of the distribution monad.}

We now give a more conceptual view of the convex structure on effect modules $A$, as an action of the probability monad $\Dist_{[0,1]}$ on the functor $\Dist_A$.

An \emph{action} of a monad $T$ on a category $\mathcal{C}$ on a functor $F:\mathcal{C}\to\mathcal{C}$ is a natural transformation
\[
\alpha : T \circ F \Rightarrow F
\]
satisfying the unit and associativity laws:
\[
\alpha_X \circ \eta_{F(X)} = \mathrm{id}_{F(X)},
\qquad
\alpha_X \circ T(\alpha_X) = \alpha_X \circ \mu_{F(X)}.
\]
Thus for each object $X$ of $\mathcal{C}$, $\alpha_X$ gives a $T$-algebra structure on $F(X)$.

In the present setting, the convex structure on an effect module $A$ induces such an action
\[
\alpha_X : \Dist_{[0,1]} (\Dist_A(X)) \to \Dist_A(X),
\]
given by convex combination of $A$-valued distributions:
\[
\alpha_X\!\left(\sum_i r_i \delta_{d_i}\right)(x)
=
\bigoplus_i r_i \cdot d_i(x).
\]
\begin{proof}[Verification of the monad action laws]
We verify that the maps
\[
\alpha_X : \Dist_{[0,1]}(\Dist_A(X)) \to \Dist_A(X)
\]
defined by
\[
\alpha_X\!\left(\sum_i r_i \delta_{d_i}\right)(x)
=
\bigoplus_i r_i \cdot d_i(x)
\]
make $\Dist_A$ into a $\Dist_{[0,1]}$-algebra--valued functor.

First, $\alpha_X$ is well-defined. Indeed, for each $x\in X$, the expression
\[
\bigoplus_i r_i \cdot d_i(x)
\]
is defined by convexity of $A$, and summing over $x$ gives
\[
\bigoplus_{x\in X} \alpha_X\!\left(\sum_i r_i \delta_{d_i}\right)(x)
=
\bigoplus_{x\in X} \bigoplus_i r_i \cdot d_i(x)
=
\bigoplus_i r_i \cdot \left(\bigoplus_{x\in X} d_i(x)\right)
=
\bigoplus_i r_i \cdot 1
=
1,
\]
so $\alpha_X(\sum_i r_i \delta_{d_i})$ is again an $A$-distribution.

For the unit law, let $d\in \Dist_A(X)$. Then
\[
\alpha_X(\eta_{\Dist_A(X)}(d))(x)
=
\alpha_X(\delta_d)(x)
=
1 \cdot d(x)
=
d(x),
\]
so
\[
\alpha_X \circ \eta_{\Dist_A(X)} = \mathrm{id}_{\Dist_A(X)}.
\]

For the associativity law, let
\[
\Phi = \sum_j s_j \delta_{\Psi_j} \in \Dist_{[0,1]}(\Dist_{[0,1]}(\Dist_A(X))),
\qquad
\Psi_j = \sum_i r_{ji}\delta_{d_{ji}}.
\]
Then
\[
\alpha_X\bigl(\Dist_{[0,1]}(\alpha_X)(\Phi)\bigr)(x)
=
\bigoplus_j s_j \cdot \alpha_X(\Psi_j)(x)
=
\bigoplus_j s_j \cdot \left(\bigoplus_i r_{ji}\cdot d_{ji}(x)\right).
\]
Using distributivity of scalar multiplication over $\oplus$ and associativity of convex combinations, this becomes
\[
\bigoplus_{j,i} (s_j r_{ji}) \cdot d_{ji}(x).
\]
On the other hand,
\[
\mu_{\Dist_A(X)}(\Phi)
=
\sum_{j,i} s_j r_{ji}\, \delta_{d_{ji}},
\]
and therefore
\[
\alpha_X\bigl(\mu_{\Dist_A(X)}(\Phi)\bigr)(x)
=
\bigoplus_{j,i} (s_j r_{ji}) \cdot d_{ji}(x).
\]
Hence
\[
\alpha_X \circ \Dist_{[0,1]}(\alpha_X)
=
\alpha_X \circ \mu_{\Dist_A(X)}.
\]

Finally, naturality in $X$ is immediate from the fact that both marginalization and convex combination are defined pointwise. 
\end{proof}

Thus for an effect module $A$, each $\Dist_A(X)$ carries a canonical $\Dist_{[0,1]}$-algebra structure, and $\Dist_A$ defines a functor
\[
\Dist_A : \mathbf{Set} \to \mathbf{Alg}(\Dist_{[0,1]}),
\]
into the category of convex sets. In particular, convex combinations of $A$-valued distributions are well-defined and natural in $X$.


\subsection{The embedding theorem for effect-valued distributions}

Given a morphism $h : A \to B$ of effect algebras, we can define a natural transformation $\Dist_h : \Dist_A \natarrow \Dist_B$ by post-composition: for $d \in \Dist_A(X)$, $\Dist_{h,X}(d) = h \circ d \in \Dist_B(X)$. Naturality is easily verified.

\begin{theorem}
The assignment
\[
A \longmapsto \Dist_A,\qquad h\longmapsto \Dist_{h}
\]
defines a full and faithful functor
\[
\mathbf{EA}\longrightarrow [\mathbf{Set},\mathbf{Set}].
\]
\end{theorem}

\begin{proof}
We show that the assignment $A \mapsto \Dist_A$, $h \mapsto \Dist_{h}$ is full and faithful.

\medskip

\noindent\emph{Faithfulness.}
Let $h,k:A\to B$ be effect-algebra morphisms such that $\Dist_{h} = \Dist_k$. Consider the two-element set $2=\{0,1\}$. Every element of $\Dist_A(2)$ is of the form $(a,a^\perp)$, so $\Dist_A(2)\cong A$ via $a \mapsto (a,a^\perp)$. Then
\[
\Dist_{h,2}(a,a^\perp) = (h(a),h(a^\perp)), \qquad
\Dist_{k,2}(a,a^\perp) = (k(a),k(a^\perp)).
\]
Equality $\Dist_{h,2} = \Dist_{k,2}$ implies $h(a)=k(a)$ for all $a\in A$, hence $h=k$. Thus the functor is faithful.

\medskip

\noindent\emph{Fullness.}
Let $\tau : \Dist_A \Rightarrow \Dist_B$ be a natural transformation. We construct an effect-algebra morphism $h:A\to B$ such that $\tau = \Dist_{h}$.

\smallskip

\noindent\emph{Step 1: Define $h$.}
For $a\in A$, set
\[
h(a) := \pi_0\bigl(\tau_2(a,a^\perp)\bigr),
\]
where $\pi_0$ denotes projection onto the first coordinate. Since $\tau_2(a,a^\perp)\in \Dist_B(2)$, it has the form $(h(a),h(a)^\perp)$.

\smallskip

\noindent\emph{Step 2: Pointwise form of $\tau$.}
Let $X$ be a set and $d\in \Dist_A(X)$. For each $x\in X$, let $\chi_x:X\to 2$ be the characteristic map of $\{x\}$. Then
\[
\Dist_A(\chi_x)(d) = (d(x), \bigoplus_{x' \neq x} d(x')) = (d(x),d(x)^\perp),
\]
using normalization and uniqueness of orthosupplementation.
By naturality,
\[
\tau_2\bigl(\Dist_A(\chi_x)(d)\bigr)
=
\Dist_{B}(\chi_x)\bigl(\tau_X(d)\bigr).
\]
Taking the first coordinate yields
\[
h(d(x)) = \tau_X(d)(x).
\]
Since this holds for all $x\in X$, we obtain
\[
\tau_X(d) = h \circ d.
\]

\smallskip

\noindent\emph{Step 3: $h$ is an effect-algebra morphism.}
We verify that $h$ preserves the effect-algebra structure.

\emph{(i) Preservation of orthogonal sum.}
Let $a \perp b$ in $A$, and set $c=(a\oplus b)^\perp$. Then $(a,b,c)\in \Dist_A(3)$. Let $m:3\to 2$ merge the first two elements. Then
\[
\Dist_A(m)(a,b,c) = (a\oplus b, c).
\]
By naturality,
\[
\tau_2(a\oplus b, c)
=
\Dist_B(m)\bigl(\tau_3(a,b,c)\bigr).
\]
Using the pointwise description,
\[
\tau_3(a,b,c) = (h(a),h(b),h(c)),
\]
so the right-hand side is
\[
(h(a)\oplus h(b),\,h(c)).
\]
The left-hand side is $(h(a\oplus b),\,h(c))$. Hence
\[
h(a\oplus b) = h(a)\oplus h(b).
\]

\emph{(ii) Preservation of orthosupplement.}
From $(a,a^\perp)\in \Dist_A(2)$ and $\tau_2(a,a^\perp)=(h(a),h(a^\perp))\in \Dist_B(2)$, we obtain
\[
h(a^\perp)=h(a)^\perp.
\]

\emph{(iii) Units.}
It follows that $h(1)=1$ and $h(0)=0$.

Thus $h$ is an effect-algebra morphism.

\smallskip

\noindent\emph{Step 4: $\tau = \Dist_{h}$.}
From Step 2, for every set $X$ and $d\in \Dist_A(X)$,
\[
\tau_X(d) = h \circ d = \Dist_{h,X}(d).
\]
Hence $\tau = \Dist_{h}$.

\medskip

Therefore the functor is full and faithful.
\end{proof}

\section{Measurement models}

The central notion in the sheaf-theoretic approach to contextuality is that of \emph{empirical model}, \ie~compatible families on the presheaf $\Dist_R \circ \ES$ over a measurement scenario $(X, \Sg, O)$. Empirical models are so-named because the data they contain is all in principle directly observable, independently of any theory. One can repeatedly perform measurements, observe outcomes, and  record the frequencies. This is essentially the basic methodology of Bell tests. This theory independence leads in turn to various forms of device-independent certification \cite{fyrillas2024certified}.

This operational reading is tightly tied to the use of the probability distribution monad $\Dist_{\Real_{\geq 0}} \cong \Dist_{[0,1]}$. 
However, mathematically we can broaden the notion of empirical model by considering distributions valued in any effect algebra. There are some compelling reasons for doing so:
\begin{itemize}
\item Firstly, as we shall see, this substantially unifies the sheaf-theoretic approach to contextuality with the generalized probabilistic theories approach to quantum foundations.
\item In particular, the notions of \emph{quantum realizability} and \emph{state} are internalised and given a structural account within the sheaf-theoretic framework.
\item This allows new features to appear, as we can distinguish \emph{internal} and \emph{pointwise} versions of several key notions, including contextuality.
\item It also provides a setting for analyzing sharp vs.~unsharp measurements, and their relationship to contextuality, offering a new perspective on the issues discussed in works such as \cite{spekkens2005contextuality,liang2011specker}.
\item  It also overcomes some limitations of studying effect algebras in isolation. As observed in \cite{gudder2010effect}, effect algebras (under mild assumptions) can always be embedded in ``classical'' structures, and so do not offer a complete analysis of quantum non-classicality. Combining effect algebras with the sheaf-theoretic structure gives full scope to capturing non-locality, contextuality, and other non-classical features.
\end{itemize}

We are thus led to define \emph{$A$-measurement models}, or just $A$-models for brevity, over an effect algebra $A$ and a measurement scenario $(X, \Sg, O)$, to be compatible families $m = \{ m_C \}_{C \in \Sg}$, where $m_C \in \Dist_A \ES(C)$, and compatibility is defined as before: when $C \subseteq C' \in \Sg$, $m_C = m_{C'} |_{C}$. The restriction operation is marginalisation over $A$-valued distributions: if $d \in \Dist_A \ES(C')$ and $C \subseteq C'$, $d |_C(s) := \bigoplus_{s' |_C = s} d(s')$.

We shall retain the term empirical models for probabilistic models over the effect algebra $[0,1]$, since it is these models which describe behaviour which is directly observable in the classical world of macroscopic observers. 


As we have seen, distributions valued in the quantum effects algebra $\Eff(\HH)$ correspond to \povm's.
We thus refer to measurement models of this form as \emph{\povm-models}. They give the basic ingredients of a quantum realization for the scenario, without specifying the  probabilities induced by a specific state.
Similarly, we refer to models over the algebra of projections $\Proj(\HH)$ as \emph{\pvm-models}.

\begin{proposition}
\label{prop:jmeas}
Let $\XX = (X, \Sg, O)$ be a measurement scenario, and $m = \{ m_C \in \Dist_{\Eff(\HH)} \ES(C) \}_{C \in \Sg}$ be a family of \povm's. The following are equivalent:
\begin{enumerate}
\item The compatibility conditions hold, so $m$ is a \povm-model.
\item For each context $C$, the \povm's $m_{\{x\}}$, $x \in C$, are jointly measurable by the \povm~$m_C$.
\end{enumerate}
\end{proposition}
\begin{remark}
\label{remark:pvms}
The same holds for \pvm~models, and in this case joint measurability implies that the \pvm's $m_{\{x\}}$ pairwise commute, and $m_C$ is their product \cite{heinosaari2008notes}.
\end{remark}

Given a measurement scenario $\XX = (X, \Sg, O)$ and an effect algebra $A$, we shall write $\Mod_{A}(\XX)$ for the set of $A$-measurement models over $\XX$.
This can be promoted from a family of sets indexed by scenarios to a fibration,  using the notion of \emph{simulation} from \cite{abramsky2019comonadic}. However, we shall not enter into this here.

\begin{remark}
Note that models over the effect algebra $\{ 0, 1 \}$ are exactly the deterministic models. Indeed, $\Dist_{\{0,1\}} \ES \simeq \ES$. Since $\ES$ is a sheaf, these models are always non-contextual.
To accommodate \emph{possibilistic} models in the sense of \cite{AbramskyBrandenburger2011}, we must weaken the notion of effect algebra by omitting the final axiom: if $a \orth 1$ then $a = 0$. We can call this generalization \emph{soft effect algebras}. Boolean algebras with disjunction as a total operation (``inclusive or'') are examples of soft effect algebras. The models over $\{ 0, 1\}$ with this total operation are exactly the possibilistic models of \cite{AbramskyBrandenburger2011}.
\end{remark}
\subsection{States as observational windows}

A state $\sg : A \to [0,1]$ induces a natural transformation $\Dist_{\sg} : \Dist_A \natarrow \Dist_{[0,1]}$. 
Given a scenario $\XX = (X, \Sg, O)$, for each $C \in \Sg$ this gives a map $\Dist_{\sg, C} : \Dist_{A}(\ES(C)) \to \Dist_{[0,1]}(\ES(C))$.
This in turn induces a map $\sg_* : \Mod_A(\XX) \to \Mod_{[0,1]}(\XX)$:
\[ \sg_*(\{ m_C \}_{C \in \Sg}) \; = \; \{ \Dist_{\sg, C}(m_C) \}_{C \in \Sg} . \]
The naturality of $\Dist_{\sg}$ ensures that this preserves compatibility.

In this way, states on $A$ transform $A$-models to empirical models in the usual Abramsky-Brandenburger sense.
Thus we can regard each state of $A$ as giving a directly observable window on $A$-models.
Different $A$-models may give rise to the
same family of observable empirical behaviours, so these windows need
not completely determine the underlying measurement structure.

This opens up something new in the landscape of measurement scenarios and models. Given any property of measurement models, there are two possible versions we can consider:
\begin{itemize}
\item The \emph{internal} version, which applies directly to $A$-models for any choice of effect algebra $A$.
\item The \emph{pointwise} version, in which we take the property for empirical models, and lift it to $A$-models by considering, for each $A$-model $m$, how the property applies to the set of empirical models $\{ \sg_*(m) \mid \sg : A \to [0,1] \}$.
\end{itemize}
The comparison between these will be a theme running through the remainder of this paper.

\subsection{$A$-realizability}

The notion of \emph{quantum realizability} of empirical models was previously handled in a somewhat ad hoc fashion in \cite{AbramskyBrandenburger2011} and its sequels. We can now give a smooth structural account. An empirical model $e \in \Mod_{[0,1]}(\XX)$ is \emph{$A$-realizable} if there is a state $\sg : A \to [0,1]$ and an $A$-model $m \in \Mod_A(\XX)$ such that $\sg_*(m) = e$. In particular, if $A = \Eff(\HH)$, then $e$ is \povm-realizable, while if $A = \Proj(\HH)$, then $e$ is \pvm-realizable.

More generally, we can say that a $B$-model $m \in \Mod_B(\XX)$ is $A$-realizable if there is a morphism $h : A \to B$ such that $m$ is in the image of $h_{\ast} : \Mod_A(\XX) \to \Mod_B(\XX)$. 

\section{Contextuality for measurement models}
\label{context:sec}

Given a measurement scenario $\XX = (X, \Sg, O)$ and an effect algebra $A$, we define contextuality for $A$-models in exactly the same way as in the standard Abramsky-Brandenburger account \cite{AbramskyBrandenburger2011}. A measurement model $m \in \Mod_A(\XX)$ is \emph{$A$-non-contextual} if there is a distribution $d \in \Dist_A(\ES(X))$ such that $d |_C = m_C$ for all $C \in \Sg$. It is \emph{$A$-contextual} if there is no such distribution.
$A$-contextuality is the internal notion as discussed in the previous section.

For the effect algebras $[0,1]$ and $\{ 0,1 \}$, these notions coincide with probabilistic and logical/possibilistic contextuality as defined in \cite{AbramskyBrandenburger2011} respectively.
For the quantum effect algebra $\Eff(\HH)$, with $m$ a \povm-model, non-contextuality is the joint measurability of the \povm's in the model, and similarly for \pvm-models.

Given a state $\sg : A \to [0,1]$, we say that $m$ is $\sg$-contextual if the empirical model $\sg_*(m)$ is contextual in the sense of \cite{AbramskyBrandenburger2011}. This is the pointwise version.

\begin{proposition}
If $m$ is $A$-non-contextual, then it is $\sg$-non-contextual for all states $\sg : A \to [0,1]$.
Thus if $m$ is $\sg$-contextual for some state $\sg$, it is $A$-contextual.
\end{proposition}
\begin{proof}
Given $d \in \Dist_A(\ES(X))$ witnessing the non-contextuality of $m$, and a state $\sg : A \to [0,1]$, the push-forward $\sg_*(d) := \Dist_{\sg, \ES(X)}(d)$ witnesses the non-contextuality of $\sg_*(m)$. 
\end{proof}

Thus internal non-contextuality  implies pointwise non-contextuality for all states.
A converse requires additional assumptions.

\subsection{Coherent classical representability}
\label{subsec:coherent}

We define $\St$ to be the hom functor $\St = \hom_{\EA}(-, [0,1])$.
Given an effect algebra $A$, $\St(A)$ is its set of states. Note that $\St(A)$ carries a convex structure inherited pointwise from that on $[0,1]$.

Evaluation induces a map
\[
\iota : A \longrightarrow [0,1]^{\St(A)}, \qquad a \mapsto (\sigma(a))_{\sigma\in \St(A)}.
\]
We say that states \emph{separate points} in $A$  if  $\iota$ is injective, so that we may regard $A$ as a subset of $[0,1]^{\St(A)}$.

This extends pointwise to distributions: for any set $X$, we obtain a map
\[
\iota_* : \Dist_A(X) \hookrightarrow \prod_{\sigma\in \St(A)} \Dist_{[0,1]}(X),
\]
given by
\[
d \mapsto (\sigma_*(d))_{\sigma\in \St(A)}.
\]

The distinction between internal and observable contextuality can be
formulated as a lifting or reconstruction prooblem: given a family of empirical models generated by states from an $A$-model $m$, can we use them to reconstruct $m$?

Using the ordered-vector-space representation theorem for effect modules \cite{GudderPulmannova1998},
we can regard $A$ as an interval $[0,u]_V$ in an ordered vector space.
Finite orthogonal sums are then ordinary vector sums subject to the order-unit bound.

\begin{lemma}[Statewise lifting of \(A\)-valued distributions]
\label{lem:statewise_lifting}
Let \(A=[0,u]_V\) be an effect module whose states separate points.
Let \(Y\) be a finite set, and let \(a\in A\).

Suppose we are given a family
\[
(p_\sigma)_{\sigma\in\St(A)},
\qquad
p_\sigma:Y\to[0,1],
\]
such that:
\begin{enumerate}
\item for every \(\sigma\in\St(A)\),
\[
\sum_{y\in Y} p_\sigma(y)=\sigma(a),
\]
\item for every \(y\in Y\), there exists \(a_y\in A\) such that
\[
p_\sigma(y)=\sigma(a_y)
\qquad
\text{for all }\sigma\in\St(A).
\]
\end{enumerate}

Then there exists a unique function
\[
d:Y\to A
\]
such that
\[
\sum_{y\in Y} d(y)=a
\]
and
\[
\sigma(d(y))=p_\sigma(y)
\qquad
\text{for all }y\in Y,\ \sigma\in\St(A).
\]
\end{lemma}

\begin{proof}
Define \(d(y)=a_y\). Then for every state \(\sigma\),
\[
\sigma\!\left(\sum_{y\in Y} d(y)\right)
=
\sum_{y\in Y}\sigma(a_y)
=
\sum_{y\in Y} p_\sigma(y)
=
\sigma(a).
\]
Since states separate points,
\[
\sum_{y\in Y} d(y)=a.
\]
The remaining condition is immediate from the definition of \(d\).

Uniqueness again follows from separation of points.
\end{proof}

\begin{remark}
Ordinary probability distributions arise as the special case \(a=u\),
while subprobability distributions of total weight \(r\in[0,1]\) arise
by taking \(a=r\cdot u\).
\end{remark}

For each state \(\sigma\in\St(A)\), evaluation gives an ordinary
empirical model
\[
\sigma_\ast(m)\in \Mod_{[0,1]}(\XX).
\]
Suppose that each observable window \(\sigma_\ast(m)\) is
non-contextual. Then for each \(\sigma\) there exists a global
probability distribution
\[
p_\sigma\in \Dist_{[0,1]}(\ES(X))
\]
such that
\[
p_\sigma|_C=\sigma_\ast(m_C)
\qquad
\text{for all } C\in\Sg.
\]

The family \((p_\sigma)_{\sigma\in\St(A)}\) is said to be
\emph{coherently representable in \(A\)} if, for every global assignment
\(t\in\ES(X)\), the coefficient function
\[
\St(A)\longrightarrow [0,1],
\qquad
\sigma\longmapsto p_\sigma(t)
\]
is represented by an element of \(A\), i.e. there exists \(a_t\in A\)
such that
\[
p_\sigma(t)=\sigma(a_t)
\qquad
\text{for all }\sigma\in\St(A).
\]

\begin{theorem}[Coherent classical representability]
\label{thm:cohclassrep}
Let \(A=[0,u]_V\) be an effect module whose states separate points, and
let \(m\in\Mod_A(\XX)\). Then the following are equivalent:
\begin{enumerate}
\item \(m\) is internally \(A\)-non-contextual.
\item The family of observable empirical models
\[
\{\sigma_\ast(m)\}_{\sigma\in\St(A)}
\]
admits a coherently representable family of classical global
explanations.
\end{enumerate}
\end{theorem}

\begin{proof}
Suppose first that \(m\) is internally \(A\)-non-contextual. Then there
exists
\[
d\in \Dist_A(\ES(X))
\]
such that
\[
d|_C=m_C
\qquad
\text{for all } C\in\Sg.
\]
For each state \(\sigma\in\St(A)\), define
\[
p_\sigma:=\sigma_*(d)\in \Dist_{[0,1]}(\ES(X)).
\]
Then
\[
p_\sigma|_C
=
\sigma_*(d|_C)
=
\sigma_*(m_C),
\]
so \(p_\sigma\) is a global classical explanation of
\(\sigma_*(m)\). Moreover, for each \(t\in\ES(X)\),
\[
p_\sigma(t)=\sigma(d(t)),
\]
so the family \((p_\sigma)_\sigma\) is coherently representable, with
representing element \(d(t)\in A\).

Conversely, suppose that the observable models
\[
\{\sigma_*(m)\}_{\sigma\in\St(A)}
\]
admit a coherently representable family of classical global
explanations
\[
(p_\sigma)_{\sigma\in\St(A)}.
\]
Thus, for each \(\sigma\), \(p_\sigma\in\Dist_{[0,1]}(\ES(X))\) and
\[
p_\sigma|_C=\sigma_*(m_C)
\qquad
\text{for all } C\in\Sg.
\]
Moreover, for each \(t\in\ES(X)\), there exists \(a_t\in A\) such that
\[
p_\sigma(t)=\sigma(a_t)
\qquad
\text{for all }\sigma\in\St(A).
\]

Apply Lemma~\ref{lem:statewise_lifting} with
\[
Y=\ES(X),
\qquad
a=u.
\]
Since each \(p_\sigma\) is normalized,
\[
\sum_{t\in\ES(X)}p_\sigma(t)=1=\sigma(u),
\]
the hypotheses of the lemma are satisfied. Hence there exists a unique
\[
d\in\Dist_A(\ES(X))
\]
such that
\[
\sigma_*(d)=p_\sigma
\qquad
\text{for all }\sigma\in\St(A).
\]

It remains only to check that \(d\) is a global section for \(m\).
For every context \(C\in\Sg\), every \(s\in\ES(C)\), and every state
\(\sigma\), we have
\[
\sigma(d|_C(s))
=
(\sigma_*(d))|_C(s)
=
p_\sigma|_C(s)
=
\sigma_*(m_C)(s)
=
\sigma(m_C(s)).
\]
Since states separate points, it follows that
\[
d|_C(s)=m_C(s).
\]
Thus
\[
d|_C=m_C
\qquad
\text{for all } C\in\Sg,
\]
so \(m\) is internally \(A\)-non-contextual.
\end{proof}

\begin{remark}
This theorem shows that observable non-contextuality state by state is
not, by itself, the same as internal \(A\)-non-contextuality. The missing
ingredient is coherence across states. Internal non-contextuality is
equivalent to the existence of classical explanations for all observable
windows which assemble into a single \(A\)-valued global explanation.
\end{remark}

\subsection{Global sections versus factorisable hidden-variable models}

In the standard probabilistic setting, the Fine--Abramsky--Brandenburger theorem \cite{fine1982hidden,AbramskyBrandenburger2011} identifies several notions:
\begin{itemize}
\item existence of a global section,
\item existence of a deterministic hidden-variable model,
\item existence of a factorisable hidden-variable model.
\end{itemize}
For empirical models valued in the ordinary probability distribution monad, these notions coincide \cite{AbramskyBrandenburger2011}.

A noteworthy feature of the present generalized setting is that this equivalence no longer holds automatically at the level of general effect-valued measurement models. The reason is structural. The definition of a global section uses only the additive structure of the effect algebra:
\[
d \in \Dist_A(\ES(X)).
\]
Thus internal contextuality is well-defined for arbitrary effect algebras.

By contrast, factorisable hidden-variable models require an additional operation combining local response weights into joint weights. In the probabilistic case this is multiplication:
\[
p(s \mid C,\lambda)
=
\prod_{x\in C} p_x(s(x)\mid x,\lambda).
\]
Similarly, the construction of a global distribution from local response functions takes the form
\[
d(t)
=
\sum_{\lambda} q(\lambda)
\prod_{x\in X} p_x(t(x)\mid x,\lambda),
\]
again using multiplication.

For a general effect algebra or effect module, no such multiplication is available. Consequently, while global sections and contextuality remain well-defined, the notion of a factorisable hidden-variable model is not available at this level of generality.

This separation is conceptually significant. It shows that the sheaf-theoretic notion of contextuality is fundamentally a local-to-global consistency condition, independent of probabilistic factorisation structure. In the ordinary probabilistic setting, factorisability and existence of global sections coincide because the semiring structure of $[0,1]$ supplies a compatible multiplication. Our more general approach separates these notions, exposing contextuality as the more primitive concept.

Extending Fine–Abramsky–Brandenburger factorisation arguments requires additional multiplicative structure on effect algebras.  In particular, \emph{effect monoids} \cite{jacobs2011probabilities} provide operations
\[
\_ \cdot \_ : A \times A\to A
\]
which allow local weights to be combined into joint weights. If we assume that we are working over a  \emph{commutative effect monoid}, we can formulate generalized factorisable hidden-variable models and recover analogues of the Fine--Abramsky--Brandenburger correspondence. The standard result from \cite{AbramskyBrandenburger2011} indeed arises from the effect monoid $[0,1]$.

Existing structure theory \cite{WesterbaanWesterbaanVanDeWetering2020,ChoWesterbaanVanDeWetering2022} shows that commutative effect monoids are highly constrained, lying close to MV-algebraic or function-algebraic settings, which are essentially classical. This suggests that the coincidence between global sections and factorisable hidden-variable models is a special feature of classical probabilistic structure, rather than a generic property of contextuality itself.


\section{Convex structure of measurement models}

\subsection{The convex structure of $\Mod_A(\XX)$}

Fix a measurement scenario $\XX = (X,\Sg,O)$ and an effect module $A$.

The set $\Mod_A(\XX)$ of $A$-models over $\XX$ carries a natural convex  structure defined pointwise. Given $m,m' \in \Mod_A(\XX)$ and $r\in[0,1]$, define
\[
(rm \oplus (1-r)m')_C(s)
:=
r\cdot m_C(s) \oplus (1-r)\cdot m'_C(s).
\]
The fact that convex combination is a well-defined total operation needs an argument.
\begin{lemma}
Let \(A\) be an effect module, and let
\[
m,m'\in \Mod_A(\XX).
\]
For every \(r\in[0,1]\), the pointwise convex combination
\[
(rm \oplus (1-r)m')_C(s)
:=
r\cdot m_C(s)\oplus (1-r)\cdot m'_C(s)
\]
is well-defined and defines an element of \(\Mod_A(\XX)\).
\end{lemma}

\begin{proof}
Fix \(C\in\Sg\) and \(s\in\ES(C)\). We first show that
\[
r\cdot m_C(s)\oplus (1-r)\cdot m'_C(s)
\]
is defined.

Since \(m'\) is normalized,
\[
\bigoplus_{t\in\ES(C)} m'_C(t)=1.
\]
Applying scalar multiplication and using distributivity,
\[
\bigoplus_{t\in\ES(C)} (1-r)\cdot m'_C(t)
=
(1-r)\cdot 1.
\]
Hence
\[
(1-r)\cdot m'_C(s)\le (1-r)\cdot 1.
\]

Similarly,
\[
r\cdot m_C(s)\le r\cdot 1.
\]

Now
\[
r\cdot 1 \oplus (1-r)\cdot 1
=
(r+(1-r))\cdot 1
=
1,
\]
so
\[
r\cdot 1 \perp (1-r)\cdot 1.
\]
Since orthogonality is downward closed with respect to the order,
it follows that
\[
r\cdot m_C(s)
\perp
(1-r)\cdot m'_C(s).
\]
Thus the required sum is defined.

We now verify normalization:
\[
\begin{aligned}
\bigoplus_{s\in\ES(C)}
(rm \oplus (1-r)m')_C(s)
&=
\bigoplus_s
\bigl(
r\cdot m_C(s)\oplus (1-r)\cdot m'_C(s)
\bigr)
\\
&=
\left(
\bigoplus_s r\cdot m_C(s)
\right)
\oplus
\left(
\bigoplus_s (1-r)\cdot m'_C(s)
\right)
\\
&=
r\cdot
\left(
\bigoplus_s m_C(s)
\right)
\oplus
(1-r)\cdot
\left(
\bigoplus_s m'_C(s)
\right)
\\
&=
r\cdot 1 \oplus (1-r)\cdot 1
\\
&=
1.
\end{aligned}
\]

Finally, compatibility is preserved because restriction maps are
defined by marginalization, which preserves finite orthogonal sums and
commutes with scalar multiplication:
\[
(rm \oplus (1-r)m')_C|_U
=
r(m_C|_U)\oplus (1-r)(m'_C|_U).
\]
Since \(m\) and \(m'\) are compatible families, so is
\[
rm\oplus(1-r)m'.
\]
Hence
\[
rm\oplus(1-r)m'\in \Mod_A(\XX).
\]
\end{proof}

Thus $\Mod_A(\XX)$ is closed under finite convex combinations, and hence forms a convex  set.

More conceptually, we have
\[
\begin{tikzcd}
\Mod_A(\XX)
\arrow[r, hookrightarrow]
&
\displaystyle \prod_{C\in\Sg} \Dist_A(\ES(C))
\arrow[r, shift left=0.7ex, "{\rho^{C}_{C \cap C'}}"]
\arrow[r, shift right=0.7ex, swap, "{\rho^{C'}_{C \cap C'}}"]
&
\displaystyle \prod_{C,C'\in\Sg} \Dist_A(\ES(C\cap C'))
\end{tikzcd}
\]
the equalizer imposing the compatibility conditions. Since products and equalizers of convex sets (with convex-linear maps) are convex, it follows that $\Mod_A(\XX)$ is a convex subset of a product of the spaces $\mathcal D_A(\ES(C))$.

\medskip

\noindent
\textbf{Examples}
\begin{itemize}
\item If $A=[0,1]$, then $\Mod_A(\XX)$ is a convex polytope (the usual no-signalling polytope).
\item If $A=\{0,1\}$, there is no nontrivial convex structure.
\item For general $A$, $\Mod_A(\XX)$ is typically not a polytope, and its geometry reflects the convex structure of $A$.
\end{itemize}

\subsection{Extreme points}

\begin{definition}
An element $x$ of a convex set $C$ is \emph{extreme} if whenever
\[
x = r x^{(1)} \oplus (1-r)x^{(2)} \qquad (0<r<1),
\]
we have $x^{(1)} = x^{(2)} = x$.
\end{definition}

There are three levels of convex sets to be considered:
\begin{itemize}
\item An effect module $A$ is itself a convex set.
\item The set of $A$-distributions $\Dist_A(X)$ for each set $X$ is convex.
\item The set of measurement models $\Mod_A(\XX)$ over a measurement scenario $\XX$.
\end{itemize}
Passing to distributions adds the normalization constraint; while passing to measurement models adds the compatibility constraints.

As a first example to show the increasing subtleties which arise, consider the classical probability effect algebra $[0,1]$:
\begin{itemize}
\item The extreme points of $[0,1]$ are $\{ 0,1 \}$.
\item The extreme points of $\Dist_{[0,1]}(X)$ are the $\delta$-distributions $\delta_x$, $x \in X$, \ie~the deterministic assignments.
\item The extreme points of $\Mod_{[0,1]}(\XX)$ for $\XX$ the $(2,2,2)$ Bell scenario include not just the 16 deterministic assignments, but also the 8 PR-boxes, which are strongly contextual empirical models \cite{AbramskyBrandenburger2011}.
\item For $(3,2,2)$ Bell scenarios, there are $53,856$ extreme points, falling into 46 distinct equivalence classes under relabelling \cite{PironioBancalScarani2011}.
\item No contextual extreme point of a no-signalling polytope is \pvm-realizable \cite{ramanathan2016no}.
\end{itemize}

We now consider the case of quantum effects $\Eff(\HH)$.

\begin{proposition}
The extreme points of $\Eff(\HH)$ are the projections.
\end{proposition}

\begin{proof}[Sketch]
If an effect $E$ has a spectrum containing values in $(0,1)$, it admits nontrivial perturbations within $[0,I]$, yielding a nontrivial convex decomposition.
\end{proof}

However, when we pass to $\Eff(\HH)$-valued distributions, \ie~\povm's, this is no longer the case. In general, extremal \povm's are not \pvm's. The following example is standard \cite{DArianoLoPrestiPerinotti2005,Pellonpaa2011}, and instructive.

\begin{example}[An extreme POVM which is not projective]
The distinction between extreme effects and extreme measurements already appears for qubits.

Let $\mathcal H=\mathbb C^2$, and choose three unit vectors
\[
\psi_1,\psi_2,\psi_3 \in \mathcal H
\]
whose Bloch vectors $n_1,n_2,n_3 \in \mathbb R^3$ lie in a common plane at angles $0$, $2\pi/3$, and $4\pi/3$. Let
\[
P_i = |\psi_i\rangle\langle \psi_i|
\]
be the corresponding rank-one projections. Writing
\[
P_i = \frac{1}{2}(I + n_i \cdot \sigma),
\]
where $\sigma=(\sigma_x,\sigma_y,\sigma_z)$ is the Pauli vector, we have
\[
n_1+n_2+n_3 = 0,
\]
and hence
\[
P_1+P_2+P_3
=
\frac{1}{2}\bigl(3I + (n_1+n_2+n_3)\cdot \sigma \bigr)
=
\frac{3}{2}I.
\]
Define
\[
E_i = \frac{2}{3}P_i, \qquad i=1,2,3.
\]
Then each $E_i$ is an effect but not a projection, and
\[
E_1+E_2+E_3 = I,
\]
so $(E_1,E_2,E_3)$ is a POVM. This is the standard trine measurement.

We claim that this POVM is extreme in the convex set of three-outcome POVMs on $\mathbb C^2$.

Suppose
\[
E_i = r F_i \oplus (1-r) G_i \qquad (0<r<1)
\]
for two POVMs $(F_i)$ and $(G_i)$. Since each $E_i$ has rank one, positivity implies that both $F_i$ and $G_i$ have support contained in the support of $E_i$, hence
\[
F_i = \alpha_i E_i, \qquad G_i = \beta_i E_i
\]
for some scalars $\alpha_i,\beta_i \geq 0$.

Using normalization,
\[
I = \sum_i F_i = \sum_i \alpha_i E_i = \frac{2}{3}\sum_i \alpha_i P_i.
\]
Substituting the Bloch form of the $P_i$, this becomes
\[
I
=
\frac{2}{3}\sum_i \alpha_i \frac{1}{2}(I+n_i\cdot\sigma)
=
\frac{1}{3}\left(\sum_i \alpha_i\right) I
+
\frac{1}{3}\left(\sum_i \alpha_i n_i\right)\cdot \sigma.
\]
Comparing coefficients of $I$ and $\sigma$, we obtain
\[
\sum_i \alpha_i = 3,
\qquad
\sum_i \alpha_i n_i = 0.
\]
Now $n_1,n_2,n_3$ are three unit vectors in a plane at angles $0,2\pi/3,4\pi/3$, so they satisfy a single linear relation, namely
\[
n_1+n_2+n_3=0.
\]
Therefore the condition $\sum_i \alpha_i n_i=0$ implies
\[
\alpha_1=\alpha_2=\alpha_3.
\]
Together with $\sum_i \alpha_i=3$, this gives
\[
\alpha_1=\alpha_2=\alpha_3=1.
\]
Hence $F_i=E_i$ for all $i$. Similarly $G_i=E_i$ for all $i$.

Therefore the trine POVM is extreme.

This example shows that, although the extreme points of the effect algebra $\Eff(\HH)$ are precisely the projections, the extreme points of the distribution space $\Dist_{\Eff(\HH)}(X)$ need not be projective measurements. Extreme POVMs may be genuinely non-projective.
\end{example}

\section{Linear representations of effect algebras and measurement models}

There is a well-developed theory of linear representations for effect algebras and effect modules \cite{GudderPulmannova1998}.
In particular, for effect modules the following result can be distilled from \cite{GudderPulmannova1998}:
\begin{theorem}
There is an equivalence of categories between $\EM$ and $\OV$, the category of real ordered vector spaces with generating order units, and unital ordered linear maps.
This represents an effect module $A$ as $[\Zero, u]_V$, the interval between $\Zero$ and $u$ in the ordered vector space $(V, V^{+})$, where $u \in V^{+}$ generates $V^+$, and $V^+$ generates $V$.
\end{theorem}

In this representation:
\begin{itemize}
\item $\oplus$ is inherited from vector addition, defined when $v+w \le u$;
\item scalar multiplication $r\cdot v$ is the usual multiplication in $V$;
\item convex combinations in $A$ are restrictions of linear combinations in $V$.
\end{itemize}

\subsection{States and duality}

In the vector-space representation, states $\sigma:A\to[0,1]$ correspond precisely to positive unital linear functionals
\[
\sigma:V\to\mathbb R, \qquad \sigma(V_+) \subseteq \mathbb R_{\ge 0}, \quad \sigma(u)=1.
\]

Thus $A$ embeds into the space of functions on its state space via evaluation, and we recover the familiar state/effect duality of ordered vector spaces and generalized probabilistic theories~\cite{janotta2014generalized,plavala2023general}.

Using the ordered-vector-space representation of convex effect
algebras, internal non-contextuality can be formulated as a cone
feasibility problem, with generalized Bell witnesses arising from the
corresponding dual cone.
The quantitative refinement of contextuality can be formulated as an optimization problem generalizing the contextual fraction \cite{AbramskyBrandenburger2011}. 

\subsection{$A$-distributions and $A$-models as cone-valued data}

Let $X$ be a finite set. An element $d\in \Dist_A(X)$ corresponds to a family $(v_x)_{x\in X}\subseteq V_+$ satisfying
\[
\sum_{x\in X} v_x = u.
\]
Thus $\mathcal D_A(X)$ is identified with the convex set
\[
\left\{\, (v_x)_{x\in X}\in V_+^X \;\middle|\; \sum_{x\in X} v_x = u \,\right\}.
\]

Fix a measurement scenario $(X,\mathcal M,O)$ and write $G=\ES(X)$ for the set of global assignments. Under the interval representation, an $A$-model $e$ consists of elements
\[
e_C(s)\in V_+, \qquad C\in \mathcal M,\ s\in \ES(C),
\]
such that
\[
\sum_{s\in \ES(C)} e_C(s)=u
\]
for every context $C$, together with the usual compatibility conditions under restriction.

An $A$-model $e$ is $A$-non-contextual iff there exists a family
\[
d(t)\in V_+, \qquad t\in G,
\]
such that
\[
\sum_{t\in G} d(t)=u
\]
and
\[
e_C(s)=\sum_{t|_C=s} d(t)
\]
for all contexts $C$ and local sections $s\in \ES(C)$. Thus internal non-contextuality is a linear feasibility problem in the ordered vector space $V$, with positivity and normalization constraints.

\subsubsection*{The marginal map}

It is useful to package the above constraints into a single linear map. Define
\[
M_V : V^G \longrightarrow \prod_{C\in\mathcal M} V^{\ES(C)}
\]
by
\[
(M_V d)_C(s) := \sum_{t|_C=s} d(t).
\]
We can write $M_V$ explicitly as a $0/1$ \emph{incidence matrix} as in \cite{AbramskyBrandenburger2011}, with dimensions $|G| \times \sum_{C \in \MM} |\ES(C)|$. We define 
\[ M_V[t,s] = \begin{cases}
      1 & \text{if  $t |_C = s$} \\
      0  & \text{otherwise.}
    \end{cases}
\]
Then an $A$-model $m$ is $A$-non-contextual iff there exists $d\in V_+^G$ such that
\[
\sum_{t\in G} d(t)=u
\qquad\text{and}\qquad
M_V(d)=m.
\]
Hence the set of non-contextual $A$-models is the image under $M_V$ of the positive normalized cone slice in $V^G$.

\subsubsection*{GPT interpretation}

This places the present approach squarely in GPT form. The positive cone $V_+$ plays the role of the cone of unnormalized effects, the order unit $u$ plays the role of the unit effect, and states are positive normalized functionals. An $A$-model is therefore a compatible family of GPT-valued measurements, and observable empirical models are obtained by evaluation under states:
\[
\sigma_*(e_C)(s)=\sigma(e_C(s)).
\]
The distinction between internal and observable contextuality is then the distinction between feasibility in the ordered vector space itself and feasibility after application of positive unital functionals.

\section{Contextual fraction for measurement models}

The ordered-vector-space formulation suggests a direct analogue of the contextual fraction  \cite{AbramskyBrandenburger2011}. In the probabilistic case $A=[0,1]$, the contextual fraction is the largest weight of a non-contextual submodel in a convex decomposition, and can be computed by linear programming \cite{AbramskyBrandenburger2011}. 

For a general effect module $A$, define the \emph{$A$-non-contextual fraction} of an $A$-model $m$ by
\[
\mathrm{NCF}_A(m)
:=
\sup\Bigl\{\, r\in[0,1] \;\Big|\; m = r m^{\mathrm{NC}} \oplus (1-r)m' \} \]
for some $A$-non-contextual model  $m^{\mathrm{NC}}$
and some $A$-model $m'$.

The corresponding \emph{$A$-contextual fraction} is
\[
\mathrm{CF}_A(m):=1-\mathrm{NCF}_A(m).
\]
This definition reduces to the usual contextual fraction when $A=[0,1]$, and is intrinsic to the convex structure of $\Mod_A(\XX)$.

\subsubsection*{Linear optimization form}

In the interval representation $A=[0,u]_V$, the quantity $\mathrm{NCF}_A(m)$ can be formulated as an optimization problem. One seeks a scalar $r\in[0,1]$ and a global family
\[
d(t)\in V_+, \qquad t\in G,
\]
such that
\[
\sum_{t\in G} d(t)= r u
\]
and
\[
\sum_{t|_C=s} d(t)\le m_C(s)
\]
for all contexts $C$ and sections $s\in \ES(C)$, where $\le$ is the order induced by the cone $V_+$.

Intuitively, $d$ specifies a subnormalized non-contextual part of $e$, of total weight $r$. The maximal such $r$ is the non-contextual fraction. In the classical case $V=\mathbb R$, this reduces to the standard linear program for the contextual fraction \cite{AbramskyBarbosaMansfield2017}. In general, it is a cone program over the ordered vector space $V$.

\subsection*{Pointwise contextual fraction profile}

Besides the internal quantity $\mathrm{CF}_A(m)$, one may also consider the \emph{statewise contextual fraction profile}
\[
\sigma \longmapsto \mathrm{CF}\bigl(\sigma_*(m)\bigr), \quad \sigma \in \St(A)
\]
where the right-hand side is the ordinary contextual fraction of the observed empirical model. This yields a function on the state space $\St(A)$.

\begin{proposition}
\label{prop:CF_domination}
Let \(A\) be an effect module, let \(\XX\) be a measurement scenario,
and let
\[
m \in \Mod_A(\XX).
\]
Then for every state \(\sigma \in \St(A)\),
\[
\CF(\sigma_*(m))
\;\leq\;
\CF_A(m).
\]
Equivalently,
\[
\CFobs(m)
:=
\sup_{\sigma\in\St(A)}
\CF(\sigma_*(m))
\;\leq\;
\CF_A(m).
\]
\end{proposition}

\begin{proof}
Suppose
\[
m
=
r\,m^{\mathrm{NC}}
+
(1-r)m'
\]
where \(m^{\mathrm{NC}}\) is \(A\)-non-contextual.
Applying the state \(\sigma\) pointwise gives
\[
\sigma_*(m)
=
r\,\sigma_*(m^{\mathrm{NC}})
+
(1-r)\sigma_*(m').
\]
By Proposition~\ref{context:sec}, since
\(m^{\mathrm{NC}}\) is internally non-contextual,
\(\sigma_*(m^{\mathrm{NC}})\) is probabilistically
non-contextual.
Hence \(r\) is a feasible weight in the definition of
\(\NCF(\sigma_*(m))\), so
\[
\NCF(\sigma_*(m))
\ge r.
\]
Taking the supremum over all internal decompositions yields
\[
\NCF(\sigma_*(m))
\ge
\NCF_A(m),
\]
and so
\[
\CF(\sigma_*(m))
\le
\CF_A(m).
\]
Taking the supremum over states gives the second statement.
\end{proof}

The converse need not hold: statewise decompositions may fail to assemble into a single $A$-valued decomposition, for exactly the same coherence reasons discussed previously in Section~\ref{context:sec} in connection with contextuality.

We now identify conditions for when the internal and statewise contextual fractions coincide. The key requirement is a coherence property for non-contextual decompositions, analogous to the coherence of classical representations from Section~\ref{subsec:coherent}. 

\begin{theorem}[Coherent characterization of the contextual fraction]
\label{thm:CF_coherent_iff}
Let \(A=[0,u]_V\) be an effect module whose states separate points, and
let \(m\in\Mod_A(\XX)\). Define
\[
\NCFobs(m)
:=
\inf_{\sigma\in\St(A)}
\NCF\bigl(\sigma_*(m)\bigr),
\]
so that
\[
\CFobs(m)=1-\NCFobs(m).
\]
Assume that the supremum defining \(\NCF_A(m)\) is attained.

Then the following are equivalent:
\begin{enumerate}
\item
\[
\CF_A(m)=\CFobs(m).
\]

\item There exists a family of subprobability distributions
\[
(p_\sigma)_{\sigma\in\St(A)},
\qquad
p_\sigma:\ES(X)\to[0,1],
\]
such that:
\begin{enumerate}
\item for every \(\sigma\in\St(A)\),
\[
\sum_{t\in\ES(X)}p_\sigma(t)=\NCFobs(m);
\]

\item for every \(\sigma\in\St(A)\), \(p_\sigma\) is a
non-contextual submodel of \(\sigma_*(m)\), i.e.
\[
p_\sigma|_C(s)\leq \sigma(m_C(s))
\qquad
\text{for all } C\in\Sg,\;s\in\ES(C);
\]

\item the family is coherently representable in \(A\): for every
\(t\in\ES(X)\), there exists \(a_t\in A\) such that
\[
p_\sigma(t)=\sigma(a_t)
\qquad
\text{for all } \sigma\in\St(A).
\]
\end{enumerate}
\end{enumerate}
\end{theorem}

\begin{proof}
We use the equivalent formulation
\[
\CF_A(m)=\CFobs(m)
\quad\Longleftrightarrow\quad
\NCF_A(m)=\NCFobs(m),
\]
where
\[
\NCFobs(m):=
\inf_{\sigma\in\St(A)}
\NCF(\sigma_*(m)).
\]
Recall that
\[
\CFobs(m)\leq \CF_A(m),
\]
and hence
\[
\NCF_A(m)\leq \NCFobs(m).
\]

Suppose first that there is a coherently representable family
\[
(p_\sigma)_{\sigma\in\St(A)}
\]
of non-contextual subdistributions of the observable models
\(\sigma_*(m)\), all of total weight \(\NCFobs(m)\). Thus
\[
\sum_{t\in\ES(X)}p_\sigma(t)=\NCFobs(m),
\]
and
\[
p_\sigma|_C(s)\leq \sigma(m_C(s))
\qquad
(C\in\Sg,\;s\in\ES(C)).
\]
Coherent representability means that for each \(t\in\ES(X)\) there is
\(a_t\in A\) such that
\[
p_\sigma(t)=\sigma(a_t)
\qquad
\text{for all }\sigma\in\St(A).
\]

Apply Lemma~\ref{lem:statewise_lifting} with
\[
Y=\ES(X),
\qquad
a=\NCFobs(m)\cdot u.
\]
We obtain a unique \(A\)-valued subdistribution
\[
d:\ES(X)\to A
\]
such that
\[
\sum_{t\in\ES(X)}d(t)=\NCFobs(m)\cdot u
\]
and
\[
\sigma(d(t))=p_\sigma(t)
\qquad
(t\in\ES(X),\;\sigma\in\St(A)).
\]

Define
\[
n_C(s):=\sum_{t|_C=s}d(t).
\]
Then \(n\) is an internally non-contextual submodel, witnessed by
\(d\), of weight \(\NCFobs(m)\). Moreover, for every state \(\sigma\),
\[
\sigma(n_C(s))
=
\sum_{t|_C=s}\sigma(d(t))
=
\sum_{t|_C=s}p_\sigma(t)
=
p_\sigma|_C(s)
\leq
\sigma(m_C(s)).
\]
Since states separate points, this implies
\[
n_C(s)\leq m_C(s).
\]
Hence \(d\) witnesses an internal non-contextual submodel of \(m\) of
weight \(\NCFobs(m)\). Therefore
\[
\NCF_A(m)\geq \NCFobs(m).
\]
Together with the general inequality
\[
\NCF_A(m)\leq \NCFobs(m),
\]
we obtain
\[
\NCF_A(m)=\NCFobs(m),
\]
and hence
\[
\CF_A(m)=\CFobs(m).
\]

Conversely, suppose
\[
\CF_A(m)=\CFobs(m),
\]
equivalently
\[
\NCF_A(m)=\NCFobs(m).
\]
By the assumed attainment of \(\NCF_A(m)\), choose an \(A\)-valued
subdistribution
\[
d:\ES(X)\to A
\]
such that
\[
\sum_{t\in\ES(X)}d(t)=\NCF_A(m)\cdot u
\]
and
\[
d|_C(s)\leq m_C(s)
\qquad
(C\in\Sg,\;s\in\ES(C)).
\]
For each state \(\sigma\), define
\[
p_\sigma:=\sigma_*(d).
\]
Then
\[
\sum_{t\in\ES(X)}p_\sigma(t)
=
\sigma\!\left(\sum_t d(t)\right)
=
\sigma(\NCF_A(m)\cdot u)
=
\NCF_A(m)
=
\NCFobs(m).
\]
Also,
\[
p_\sigma|_C(s)
=
\sigma(d|_C(s))
\leq
\sigma(m_C(s)),
\]
so \(p_\sigma\) is a non-contextual subdistribution of
\(\sigma_*(m)\) of weight \(\NCFobs(m)\). Finally, the family is
coherently representable by taking
\[
a_t=d(t).
\]
Thus the required coherent family exists.
\end{proof}

\begin{remark}
The uniform weight in the theorem is not a statewise optimality
requirement. For states \(\sigma\) with
\[
\NCF(\sigma_*(m))>\NCFobs(m),
\]
the subdistribution \(p_\sigma\) need not be optimal. The theorem
requires only that the worst-case observable non-contextual weight can
be realised coherently across all states.
\end{remark}

Thus the gap between the two notions is precisely a failure of coherent representability across states.

\section{Cone duality and generalized Bell witnesses}

The cone-program formulation of contextuality admits a natural dual
description generalizing Bell inequalities.

Let \(A=[0,u]_V\) be an effect module represented as the unit interval in
an ordered vector space \((V,V_+)\), and let
\[
G=\ES(X)
\]
be the set of global assignments for a measurement scenario
\(\XX=(X,\mathcal M,O)\).

Recall that an \(A\)-model \(m\) is \(A\)-non-contextual iff there exists
\[
d\in V_+^G
\]
such that
\[
\sum_{t\in G} d(t)=u
\]
and
\[
(M_V d)_C(s)
=
\sum_{t|_C=s} d(t)
=
m_C(s).
\]

The \(A\)-non-contextual fraction can be formulated as the cone programme
\[
\mathrm{NCF}_A(m)
=
\sup
\left\{
\lambda\in[0,1]
\;\middle|\;
\exists d\in V_+^G,\;
M_V d \le m,\;
\sum_{t\in G} d(t)=\lambda u
\right\}.
\]

The dual cone \(V_+^\ast\subseteq V^\ast\) consists of the positive linear
functionals on \(V\):
\[
V_+^\ast
=
\{
\varphi:V\to\mathbb R
\mid
\varphi(v)\ge 0
\;\;
\forall v\in V_+
\}.
\]

Applying cone duality yields generalized Bell witnesses.
These consist of families
\[
\beta_{C,s}\in V_+^\ast,
\qquad
C\in\mathcal M,\;
s\in\ES(C),
\]
which define a linear functional on measurement models:
\[
\langle \beta,m\rangle
:=
\sum_{C\in\mathcal M}
\sum_{s\in\ES(C)}
\beta_{C,s}(m_C(s)).
\]

Such a family defines a valid generalized Bell inequality if there exists
\(\alpha\in\mathbb R\) such that
\[
\langle \beta,m\rangle
\le \alpha
\]
for every \(A\)-non-contextual model \(m\).

Thus contextuality witnesses arise as elements of the dual ordered
structure:
\[
\boxed{
\text{generalized Bell witnesses}
=
V_+^\ast\text{-valued Bell functionals}.
}
\]

In the classical probabilistic case \(A=[0,1]\), we have
\[
V=\mathbb R,
\qquad
V_+^\ast=\mathbb R_{\ge 0},
\]
and the above reduces to the usual LP duality description of Bell
inequalities and contextuality inequalities.

For quantum effects,
\[
A=\Eff(\mathcal H),
\]
the ordered vector space is the self-adjoint part of
\(B(\mathcal H)\), and 
the dual cone consists of positive linear functionals, equivalently
positive operators via the trace pairing.

Generalized Bell witnesses therefore take the form
\[
\sum_{C,s}
\mathrm{Tr}(B_{C,s}\,m_C(s)),
\qquad
B_{C,s}\ge 0.
\]

These witnesses detect contextuality internally at the level of
effect-valued measurement models, before evaluation under states.
Ordinary Bell inequalities arise only after applying a state
\[
\sigma:A\to[0,1],
\]
which maps the generalized witness to a scalar-valued inequality for the
observable empirical model \(\sigma_\ast(m)\).

\begin{example}[Tsirelson operators as internal Bell witnesses]
Consider the CHSH Bell scenario with binary observables
\[
A_0,A_1
\qquad\text{and}\qquad
B_0,B_1,
\]
taking values in \(\{\pm1\}\).

In the standard probabilistic setting, the CHSH Bell functional is
\[
\langle \mathcal B \rangle
=
\langle A_0B_0\rangle
+
\langle A_0B_1\rangle
+
\langle A_1B_0\rangle
-
\langle A_1B_1\rangle,
\]
and classical hidden-variable models satisfy the Bell inequality
\[
|\langle \mathcal B\rangle|
\le 2.
\]

In quantum theory, one associates to these observables the CHSH operator
\[
\mathcal B
=
A_0\otimes B_0
+
A_0\otimes B_1
+
A_1\otimes B_0
-
A_1\otimes B_1.
\]
Tsirelson's theorem states that
\[
\|\mathcal B\|
\le 2\sqrt 2,
\]
and this bound is attained.

From the present perspective, the important point is that
\(\mathcal B\) is naturally an \emph{internal} contextuality witness at
the operator level, prior to evaluation under states.

Indeed, in the effect-valued framework, a measurement model assigns
effects
\[
m_C(s)\in \Eff(\HH)
\]
to local events \((C,s)\). Cone duality yields generalized Bell
witnesses given by positive linear functionals on the ordered vector
space of effects, equivalently by positive operators via the trace
pairing. Thus a generalized Bell functional takes the form
\[
\Phi_W(m)
=
\sum_{C,s}
\Tr(W_{C,s}\,m_C(s)),
\qquad
W_{C,s}\ge 0.
\]

The CHSH operator may be viewed as a special case of such an internal
witness. Given a quantum state \(\rho\), evaluation under the Born rule
yields the observable quantity
\[
\Tr(\rho \mathcal B),
\]
recovering the usual CHSH expectation value.

Thus ordinary Bell inequalities arise only after applying states to
internal operator-valued witnesses. In this sense, Tsirelson-type
operators are the internal counterparts of Bell inequalities in the
effect-valued setting.
\end{example}

We now show how a core result on the contextual fraction for probabilistic empirical models from \cite{AbramskyBarbosaMansfield2017} generalizes to effect-valued models.

\begin{theorem}[Dual characterization of the \(A\)-contextual fraction]
Let \(A=[0,u]_V\) be an effect module represented in an ordered real
vector space \((V,V_+)\), and let \(m\in\Mod_A(\XX)\). Suppose that the
cone programme defining \(\NCF_A(m)\) satisfies strong duality and that
the dual optimum is attained. Then there exists a feasible dual witness
\((\alpha,\beta)\) such that
\[
1-
\sum_{C,s}\beta_{C,s}(m_C(s))
=
\CF_A(m).
\]
Equivalently, the \(A\)-contextual fraction of \(m\) is exactly the
maximal violation of a generalized Bell inequality over the dual cone.
\end{theorem}

\begin{proof}
Recall that \(\NCF_A(m)\) is defined by the conic programme
\[
\NCF_A(m)
=
\sup
\left\{
\lambda
\;\middle|\;
d\in V_+^{\ES(X)},\;
M_Vd\leq m,\;
\sum_{t\in\ES(X)}d(t)=\lambda u
\right\}.
\]
Here \(M_V\) is the marginal map
\[
(M_Vd)_C(s)=\sum_{t|_C=s}d(t).
\]

The corresponding dual programme has variables
\[
\beta_{C,s}\in V_+^*,
\qquad
\alpha\in V^*,
\]
and is
\[
\inf
\sum_{C,s}\beta_{C,s}(m_C(s))
\]
subject to
\[
\alpha(u)=1
\]
and, for every global assignment \(t\in\ES(X)\),
\[
\sum_{C\in\MM}\beta_{C,t|_C}-\alpha\in V_+^*.
\]

By the assumed strong duality and dual attainment, there is a feasible
dual solution \((\alpha,\beta)\) such that
\[
\sum_{C,s}\beta_{C,s}(m_C(s))
=
\NCF_A(m).
\]

We now interpret this dual solution as a generalized Bell witness.
Let \(n\in\Mod_A(\XX)\) be internally non-contextual. Then there exists
\[
d\in\Dist_A(\ES(X))
\]
such that
\[
M_Vd=n
\qquad\text{and}\qquad
\sum_{t\in\ES(X)}d(t)=u.
\]
Using dual feasibility, we compute
\[
\begin{aligned}
\sum_{C,s}\beta_{C,s}(n_C(s))
&=
\sum_{C,s}\beta_{C,s}\bigl((M_Vd)_C(s)\bigr)\\
&=
\sum_{t\in\ES(X)}
\left(\sum_{C\in\MM}\beta_{C,t|_C}\right)(d(t))\\
&\geq
\sum_{t\in\ES(X)}\alpha(d(t))\\
&=
\alpha\!\left(\sum_{t\in\ES(X)}d(t)\right)\\
&=
\alpha(u)\\
&=
1.
\end{aligned}
\]
Thus every internally non-contextual \(A\)-model satisfies the generalized
Bell inequality
\[
\sum_{C,s}\beta_{C,s}(n_C(s))\geq 1.
\]

For the given model \(m\), optimality gives
\[
\sum_{C,s}\beta_{C,s}(m_C(s))
=
\NCF_A(m).
\]
Hence the violation of this inequality by \(m\), measured in the
normalization used above, is
\[
1-
\sum_{C,s}\beta_{C,s}(m_C(s))
=
1-\NCF_A(m)
=
\CF_A(m).
\]
Therefore \(m\) has a generalized Bell witness whose violation exactly
saturates its \(A\)-contextual fraction.
\end{proof}

\begin{remark}
For \(A=[0,1]\), we have \(V=\mathbb R\) and
\(V_+^*=\mathbb R_{\geq 0}\). The cone programme is then an ordinary
linear programme, and the theorem specializes to the duality theorem for
the contextual fraction \cite{AbramskyBarbosaMansfield2017}: every empirical model has a Bell inequality
whose normalized violation is exactly its contextual fraction.
\end{remark}

\subsection*{Discussion}
In finite-dimensional polyhedral cases, the required strong duality
hypothesis is automatic. For general ordered vector spaces or
non-polyhedral cones, such as positive semidefinite cones, one should
state the usual conic regularity assumptions, or else work with the
closed conic hull/facial reduction. Under these standard hypotheses the
dual witnesses above are the precise effect-valued analogue of Bell
inequalities.

In the finite-dimensional quantum case \(A=\Eff(\mathcal H)\), the
ordered vector space is \(B(\mathcal H)_{\mathrm{sa}}\) with positive
semidefinite cone. The cone programme for \(\NCF_A\) is therefore a
semidefinite programme. Since the cone is closed and finite-dimensional,
and the feasible set is compact, optima are attained; under the standard
Slater/regularity hypotheses of semidefinite programming, strong duality
holds. Hence the dual-witness theorem applies, yielding
operator-valued Bell witnesses saturating \(\CF_A(m)\).

In the infinite-dimensional quantum case, the ordered vector space
\(B(\mathcal H)_{\mathrm{sa}}\) is no longer finite-dimensional, and
compactness arguments used in semidefinite programming no longer apply.
The dual cone is naturally identified with the positive trace-class
operators via the trace pairing. Consequently the dual witnesses become
trace-class operator-valued Bell functionals.

Under suitable topological regularity conditions ensuring strong conic
duality and attainment --- for example weak-* closedness together with
appropriate Slater-type conditions --- the dual-witness theorem extends
to this setting.

\begin{remark}
The constructions of contextual fractions and dual Bell witnesses require
convex structure, and hence apply to effect modules rather than general
effect algebras.

In particular, the projection effect algebra
\(\Proj(\mathcal H)\) is not convex: if \(P\) is a nontrivial
projection and \(0<r<1\), then \(rP\) is not a projection.
Consequently, while \(\Proj(\mathcal H)\)-valued measurement models and
their contextuality are perfectly well-defined, there is no intrinsic
notion of convex decomposition or contextual fraction for such models.

This sharply separates the additive/local-to-global aspects of the
theory, which depend only on the effect-algebra structure, from the
convex-analytic aspects, which require effect modules.
\end{remark}

\section{Sharpness and dilation}

The main motivation for the introduction of effect algebras was to allow for unsharp effects  and measurements. In the classical case, probabilistic and fuzzy predicates generalize the ``crisp'' ones familiar from logic, while in the quantum case, quantum effects generalize projections, and \povm's generalize \pvm's.

We define an element $a$ of an effect algebra $A$ to be \emph{sharp} if $a \wedge a^\bot = 0$. (We are not assuming that meets exist in general, but rather asserting that the meet exists and is equal to $0$ in this case.)
In the case that the partial order on the effect algebra is a lattice, the sharp elements form an orthomodular lattice \cite{FoulisBennett1994}, recovering the connection with Birkhoff--von Neumann style quantum logic.

In the case of the quantum effects $\Eff(\HH)$, the sharp elements are the projections $\Proj(\HH)$. Moreover, this effect algebra is \emph{sharply dominating} \cite{gudder1998s}, meaning that every element has a smallest sharp element above it in the order.
For further results on sharp elements in effect algebras, see e.g.~\cite{rievcanova2001sharp}.

We shall look at more global questions, at the level of measurement models.
In particular, we are motivated by the following question:
\begin{question}
Which empirical models are physically realizable in quantum theory?
\end{question}
Most work on contextuality has focussed on realization by \pvm-models. The obvious generalization to accommodate unsharp measurements would be to consider realization by \povm-models. However, we have the following simple observation:
\begin{proposition}
Let $A$ be an effect module. Every empirical model is $A$-realizable with respect to any state $\sigma : A \to [0,1]$.
In particular, every empirical model is realized by a \povm-model in any Hilbert space dimension.
\end{proposition}
\begin{proof}
Given an empirical model $e$, define an $A$-model $m$ by $m_{C}(s) := e_{C}(s) \cdot 1$. Compatibility of $m$ is easily seen to follow from that of $e$. Given any state $\sigma$, using the fact that $\sigma$ is a morphism of effect modules:
\[ \sigma(m_C(s)) = \sigma(e_C(s) \cdot 1) = e_C(s) \cdot \sigma(1) = e_C(s) \cdot 1 = e_C(s) . \]
\end{proof}
The case for \povm-models is well-known in the literature \cite{liang2011specker,yu2013quantum,kunjwal2014minimal}, albeit usually phrased somewhat differently.
While realizability by a \pvm-model is a highly non-trivial condition (e.g.~it is undecidable if the Hilbert space dimension is not fixed \cite{slofstra2019set,abramsky2026team}), realizability by a \povm-model allows us to ``cheat'' by putting all the information in the convex structure. (Operationally, this amounts to ``stochastic pre-processing''. Here, this pre-processing is the whole game; what follows is a no-op.)
Thus additional constraints are needed to filter out these trivial realizations.

\subsection{Dilation in measurement models}

At first sight, the fact that there are trivial \povm-realizations where no \pvm-realizations exist seems puzzling, in the light of well-known results on dilation.

The basic result for \povm's (Naimark dilation \cite{busch2016dilation}) can be stated  as follows:
\begin{theorem}
Let $R \in \Dist_{\Eff(\HH)}(X)$ be a \povm. There is an isometry $V : \HH \to \KK$ and a \pvm~$P \in \Dist_{\Proj(\KK)}(X)$ such that
\[ R(x) = V^* P(x) V, \quad x \in X . \]
\end{theorem}
As an immediate consequence, for any state $\sigma \in \St(\Eff(\HH))$, there is a state $\sigma' \in \St(\Proj(\KK))$ such that for all $x \in X$, $\sigma(R(x)) = \sigma'(P(x))$.
Indeed, by Theorems~\ref{th:Gleason} , \ref{th:Busch}, $\sigma$ is induced by a density operator $\rho$ on $\HH$, and
$\sigma(R(x)) = \Tr(\rho R(x)) = \Tr(\rho V^* P(x) V) = \Tr(V \rho V^* P(x))$, so we can take $\sigma'$ to be the state induced by $V \rho V^*$. Thus at the level of individual measurements, the statistics produced by any \povm~on a given state can be reproduced by a \pvm~and a  suitable state on a larger Hilbert space.

This can be extended to multiple measurements. We have the following generalization  \cite{paulsen2016entanglement}:
\begin{theorem}
\label{th:multidilations}
Let $R_1, \ldots , R_n \in \Dist_{\Eff(\HH)}(X)$ be  \povm's. There is an isometry $V : \HH \to \KK$, and \pvm's~$P_1, \ldots , P_n \in \Dist_{\Proj(\KK)}(X)$, such that
\[ R_i(x) = V^* P_i(x) V, \quad x \in X, \;\; i = 1, \dots , n . \]
\end{theorem}
This result is stated for convenience for the case where all the \povm's have the same set of outcomes. It can be extended to the case where each $R_i$ has outcome set $X_i$, by forming the common outcome set $X := \bigcup_i X_i$, and extending each \povm~$R_i$ to $X$ by assigning $0$ to all outcomes $x \in X \setminus X_i$.

We could try to apply this result to \povm-models $\{ R_C \}_{C \in \Sg}$ in $\Mod_{\Eff(\HH)}(\XX)$.  We would indeed obtain a family of \pvm's $\{ P_C \}_{C \in \Sg}$, where $P_C \in \Dist_{\Proj(\KK)}(\ES(C))$, $C \in \Sg$.
However, \emph{there is no reason for this family to be compatible}, so we do not in general obtain a \povm-model. 
In particular, while we do have the equations, for contexts $C$, $x \in C$, $o \in O_x$:
\[ V^* P_{x,o} V \; = \; \bigoplus_{s \in \ES(C), s(x) = o} V^* P_C(s) V , \]
we cannot infer from this that 
\[ P_{x,o} \; = \; \bigoplus_{s \in \ES(C), s(x) = o} P_C(s) . \]
Once again, considering the global level of measurement models exposes a requirement for uniformity across contexts which does not appear when considering individual \povm's, or even sets of unrelated \povm's.

\subsection{Uniform Dilation}

Let $R = \{ R_C \}_{C \in \Sg}$ in $\Mod_{\Eff(\HH)}(\XX)$ be a \povm-model. We say that $R$ has a \emph{uniform dilation} if there is an isometry $V : \HH \to \KK$, and \pvm-model $P = \{ P_C \}_{C \in \Sg}$ in $\Mod_{\Proj(\KK)}(\XX)$, such that
\[ R_C(s) = V^* P_C(s) V, \quad C \in \Sg, \, s \in \ES(C). \]
Recall from Remark~\ref{remark:pvms} that $P_C(s) = \prod_{x \in C} P_{x, s(x)}$, a product of commuting projections.

Given such a uniform dilation, for any density operator $\rho$ inducing a state $\sigma \in \St(\Eff(\HH))$, the empirical model $e = \sigma_*(R) = \sigma'_*(P)$, where $\sigma'$ is induced by $V \rho V^*$.
Thus if a uniform dilation exists for a \povm-model, any empirical model it realizes is also realized by a \pvm-model.

Whereas the standard dilation theorems show that dilations always exist if no uniformity conditions are imposed, we can see that the non-existence of \pvm~realizations show obstructions to the existence of uniform dilations.
A basic example is provided by the well-known Specker triangle \cite{liang2011specker}.

\begin{proposition}[Specker obstruction to uniform dilation]
Let \(M_1,M_2,M_3\) be three binary POVMs which are pairwise jointly
measurable but not triplewise jointly measurable \cite{liang2011specker}. Consider the
measurement scenario whose contexts are the three pairs
\[
\{1,2\},\qquad \{2,3\},\qquad \{1,3\}.
\]
Then the corresponding  \povm-model does not admit a uniform
dilation.
\end{proposition}

\begin{proof}
Suppose a uniform dilation existed. Then each \(M_i\) would be
represented by projections \(P_{i,0},P_{i,1}\), independent of context.
Since each pair \(\{i,j\}\) is a context, the projections associated with
\(M_i\) and \(M_j\) commute for every pair \(i,j\).

Thus the three dilated PVMs commute pairwise. For projections, pairwise
commutativity implies joint commutativity, so the three PVMs admit a
single joint PVM. Compressing this joint PVM back to the original Hilbert
space yields a triplewise joint POVM for \(M_1,M_2,M_3\), contradicting
the assumption that the POVMs are not triplewise jointly measurable.
\end{proof}

This example illustrates that uniform sharp dilatability is strictly stronger
than ordinary joint measurability. It includes all \pvm-models, but
excludes \povm-models whose compatibility structure cannot arise from
context-independent sharp representatives.

Uniform dilatability expresses a consistency requirement on the
physical implementation of measurements. It requires that each
measurement be represented by a context-independent sharp observable in
an extended system, with joint measurements arising from the combination
of commuting observables. This rules out models in which compatibility
relations or correlations are introduced in a context-dependent or
ad hoc manner, and captures the idea that measurement structure should
derive from a fixed underlying physical realisation.

Uniform dilatability should be understood as a structural
constraint rather than a fundamental physical principle. While it is
satisfied by standard quantum models based on projective measurements,
general \povm-models need not admit such a representation. The condition
therefore singles out a class of measurement models with a particularly
rigid and physically interpretable implementation.

\subsection{Dilation in Bell scenarios}

In the case of non-locality, \ie contextuality over Bell scenarios, the analysis simplifies considerably. Bell scenarios have a natural notion of locality, where the scenario can be decomposed into $n$ ``sites'' or ``agents'' (``Alice'', ``Bob'', etc.), who in the intended interpretation are causally independent/spacelike-separated. A quantum realization which captures this feature will have the property that any pair of measurements $M_i$ in site $i$ and $M_j$ in site $j$, $i \neq j$, will commute.
We say that \povm-models of this form \emph{have local measurements}, or are \emph{commuting operator models}.

A particular case of this is where the underlying Hilbert space is $\HH = \HH_1 \otimes \cdots \otimes \HH_n$, with a tensor factor for each site. The measurements will correspondingly be of the form $M_1 \otimes \cdots \otimes M_n$, formed as the tensor product of local measurements. We shall refer to models of this form as \emph{tensor product models}.

In finite dimensions, by Tsirelson's lemma \cite{tsirelson2006bell}, without loss of generality we can assume that the models are of tensor product form. In infinite dimensions, this is no longer the case \cite{ji2021mip}.

\begin{proposition}
For tensor product \povm-models  over Bell scenarios, uniform dilations always exist.
\end{proposition}
\begin{proof}
We can  apply the standard Dilation Theorem~\ref{th:multidilations} to the measurements in each site, obtaining dilations $V_i : \HH_i \to \KK_i$. We then combine these as $V_1 \otimes \cdots \otimes V_n : \HH_1 \otimes \cdots \otimes \HH_n \to \KK_1 \otimes \cdots \otimes \KK_n$. It is straightforward to verify that the local projective measurements provided by the Dilation Theorem form a compatible family (see \cite{AbramskyBrandenburger2011}).
\end{proof}

As an immediate corollary, we obtain:
\begin{proposition}
All empirical models on Bell scenarios realized by tensor product \povm-models   are realized by tensor product \pvm-models.
\end{proposition}

Thus for example, while PR-boxes have trivial \povm-realizations, these realizations necessarily have non-local measurements.

\begin{remark}
In the commuting-operator setting, it is known that POVM and PVM realizations yield the same correlation sets, using operator-algebraic dilation techniques based on maximal tensor products and Stinespring dilation \cite{Paulsen2002,PaulsenSeveriniStahlke2016}. Our approach suggests a stronger structural question: whether every  commuting-operator POVM-model admits a uniform sharp dilation as a compatible measurement model in the sense developed here.
\end{remark}

\subsection{Outlook}

We have conducted our discussion of uniform dilation at the concrete level of \povm-models. This allowed us to use the dagger structure on Hilbert spaces to express isometries, and, more importantly, to use the tracial nature of states to transfer the results to empirical models. It would be interesting to see to what extent this can be lifted to the general level of measurement models, with a suitable abstract notion of ``uniform sharp dilation''.

\section{Monoidal structure and resources}
\label{sec:monstruct_resources}

As we saw in Section~\ref{subsec:failmonad}, effect algebras and modules give rise to distribution functors but not to monads. However, there is a route to obtaining \emph{graded monads} which brings important ideas concerning resources and compositionality into play, leading to connections with quantum advantage in non-local games, quantum homomorphisms of graphs and relational structures, and quantum CSP's. This was already developed in a specific case in the \emph{quantum monad on relational structures} \cite{DBLP:conf/mfcs/AbramskyBSZ17}.

The clue here is that the basic quantum examples of effect algebras, $\Eff(\HH)$ and $\Proj(\HH)$, are \emph{indexed} by the underlying Hilbert space $\HH$. The idea is to take this indexing seriously, and to exploit the monoidal structure on Hilbert spaces.
It is well-established that monoidal categories form a mathematical basis for resource theories in physics \cite{coecke2016mathematical,fritz2017resource}.

To simplify the discussion, we assume we have a monoid of resources $(\RR, {\otimes}, I)$, rather than a general monoidal category. An 
\emph{$\RR$-graded effect monoid} is an assignment $\RR \to \EA$, $R \mapsto A_R$, together with a graded family of maps $\_ \otimes_{R,S} \_ : A_R \times A_S \to A_{R \otimes S}$, such that $A_I = \{ 0,1 \}$,  and:
\begin{enumerate}
\item $(a \otimes_{R,S} b) \otimes_{R \otimes S,T} c \; = \; a \otimes_{R, S \otimes T} (b \otimes_{S,T} c), \quad a\otimes_{R,I} 1 = a = 1 \otimes_{I,R} a$
\item If $a \orth b$ in $A_R$ then $(a \otimes_{R,S} c) \orth (b \otimes_{R,S} c)$ in $A_{R \otimes S}$ and symmetrically, and
\[ (a \oplus b) \otimes_{R,S} c \; = \; (a \otimes_{R,S} c) \oplus (b \otimes_{R,S} c), \quad a  \otimes_{R,S} (b \oplus c) \; = \; (a \otimes_{R,S} b) \oplus (a \otimes_{R,S} c) \]
\item $a \otimes_{R,S} 0 \; = \; 0  \; = \; 0 \otimes_{R,S} a$.
\end{enumerate}

\subsection*{Examples}
For the main quantum examples, we use the category of complex matrices, with natural numbers as objects. This is a monoidal skeleton of the category of finite-dimensional Hilbert spaces. Tensor on objects is multiplication. For each $\mathbf{n}$, $A_{\mathbf{n}} = \Eff(\Complex^n)$, where the effects are represented as matrices with respect to the standard basis. The graded product $E \otimes_{\mathbf{n}, \mathbf{m}} E'$ is the Kronecker product of matrices. This gives an $(\Nat, {\times}, 1)$-graded effect monoid.

We can similarly exhibit a graded structure over $(\Nat, {\times}, 1)$ for the projections, as is done implicitly in \cite{DBLP:conf/mfcs/AbramskyBSZ17}.

We now show how to construct an $\RR$-graded monad on $\Set$ from an $\RR$-graded effect monoid:
\begin{itemize}
\item For each $R \in \RR$, $\Dist_{A_R}$ is an endofunctor on $\Set$, as defined in Section~\ref{sec:eff-valueddists}.
\item The natural transformation $\eta^{I}_{X} : X \to \Dist_{\{ 0,1\}}(X)$ is also defined as in Section~\ref{sec:eff-valueddists}.
\item Given $R, S \in \RR$, the graded monad multiplication $\mu_{X}^{R,S} : \Dist_{A_R}(\Dist_{A_S}(X)) \to \Dist_{A_{R \otimes S}}(X)$ is defined by
\[ \mu_{X}^{R,S}(\Phi)(x) \; = \;  \bigoplus_{d \in \Dist_{A_S}(X)} \Phi(d) \otimes_{R,S} d(x) . \]
\end{itemize}
\begin{proposition}
The above data defines an $\RR$-graded monad on $\Set$.
\end{proposition}
\begin{proof}
The proof follows exactly the same lines as that given in \cite{DBLP:conf/mfcs/AbramskyBSZ17} for the quantum monad, which only uses the generic properties of an $\RR$-graded effect monoid.
\end{proof}

The further developments in \cite{DBLP:conf/mfcs/AbramskyBSZ17} show how the quantum monad can be used to study quantum advantage in contextuality, non-local games, quantum homomorphisms and CSP. The emphasis in \cite{DBLP:conf/mfcs/AbramskyBSZ17} is on perfect strategies. The generalization to effect-valued distributions and measurement models provides a natural setting for studying quantitative aspects.

\section{Final remarks}

The theory developed here generalizes the sheaf-theoretic treatment
of contextuality from probability-valued empirical models to
effect-valued measurement models. This separates internal measurement
structure from observable probabilistic behaviour, and reveals a new
distinction between internal and observable contextuality. A basic
structural result identifies internal non-contextuality with the
existence of coherent classical explanations across all observable
windows.

The resulting theory connects sheaf-theoretic contextuality, effect
algebras, generalized probabilistic theories, and convex-geometric
methods within a common setting. In particular, contextuality acquires a
natural cone-theoretic formulation, contextual fraction extends to
effect-valued models, and sharp realizability appears as a coherence
condition on dilation structure.

More broadly, the framework separates the local-to-global structure of
contextuality from specifically probabilistic notions such as
factorisability, while also supporting compositional extensions via
graded monadic constructions. This suggests a basis for further
developments linking contextuality, quantum resources, generalized
measurements, and cohomological obstructions.

We mention some specific directions for future work.
\begin{itemize}
\item Extending cohomological characterizations of contextuality \cite{DBLP:journals/corr/abs-1111-3620,DBLP:conf/csl/AbramskyBKLM15,okay2017topological} to measurement models. The work by Roumen on cohomology of effect algebras \cite{roumen2016cohomology} should be relevant here.
\item Developing a resource theory for contextuality, generalizing \cite{AbramskyBarbosaMansfield2017,abramsky2019comonadic}.
\item Combining causality and contextuality at the level of measurement models, generalizing \cite{abramsky2024combining}.
\item Developing Vorob'ev-type theorems \cite{Vorob1962consistent} for measurement models.
\end{itemize}

\begin{acknowledgements}
Discussions with Martti Karvonen and Amin Karamlou provided valuable stimulus for this work.

Support from EPSRC Fellowship EP/V040944/1 ``Resources in Computation'' and the UCL Centre for Doctoral Training (CDT) in Quantum Computation and Quantum Communications is gratefully acknowledged.
\end{acknowledgements}

\bibliographystyle{amsplain}
\bibliography{bpbib}

\end{document}